\documentclass[aps,pra,reprint,superscriptaddress,nofootinbib]{revtex4-2}

\usepackage{amsmath,amssymb,mathtools,bm}
\usepackage{graphicx}
\usepackage{booktabs}
\usepackage[table]{xcolor}
\usepackage[colorlinks=true,citecolor=blue,linkcolor=magenta]{hyperref}
\usepackage{microtype}
\usepackage{colortbl,hhline}

\usepackage{amsmath,amssymb,graphicx}
\usepackage{multirow}

\usepackage{tabularray}
\usepackage{bm}

\usepackage{array}
\usepackage{longtable}

\usepackage{amsthm}
\usepackage{amsmath}
\usepackage[caption=false]{subfig}
\usepackage{colortbl,multirow,hhline,longtable}
\newsavebox{\SIbox}
\newcommand{\SIcell}[2]{%
  \sbox{\SIbox}{$\textstyle #2$}%
  \raisebox{0pt}[\dimexpr\ht\SIbox+1.8pt\relax][\dimexpr\dp\SIbox+1.8pt\relax]{%
  \ifdim\wd\SIbox>#1\relax
    \resizebox{#1}{!}{\usebox{\SIbox}}%
  \else\usebox{\SIbox}\fi}}

\newtheorem{theorem}{Theorem}

\newtheorem{proposition}{Proposition}

\newtheorem{definition}{Definition}

\newcommand{\ket}[1]{\left|#1\right\rangle}
\newcommand{\bra}[1]{\left\langle#1\right|}

\newcommand{\rad}{\rm rad}

\newcommand{\C}{\mathbb C}
\newcommand{\id}{\mathbb I}
\newcommand{\rank}{\rm{rank}}

\newcommand{\Span}{\rm{span}_{\C}}
\newcommand{\slc}{\mathfrak{sl}}

\newcommand{\pstab}{\mathfrak p}

\begin{document}
\title{Classification of quantum states through conserved operators}

\author{Federico Petrovich}
\affiliation{Information Sciences, Los Alamos National Laboratory, Los Alamos, New Mexico 87545, USA}

\author{R. Rossignoli}
\affiliation{Instituto de F\'isica La Plata, CONICET, and Depto.\ de F\'isica, Facultad de Ciencias Exactas, Universidad Nacional de La Plata, C.C. 67, La Plata (1900), Argentina}
\affiliation{Comisi\'on de Investigaciones Cient\'ificas (CIC), La Plata (1900), Argentina}

\author{M. Cerezo}
\affiliation{Information Sciences, Los Alamos National Laboratory, Los Alamos, New Mexico 87545, USA}

\begin{abstract}

Determining whether quantum states can be interconverted by a prescribed set of transformations is a fundamental problem in quantum information. Here we develop a general classification framework based on conserved operators (COs), i.e., not necessarily hermitian matrices  having the state of interest as an eigenstate. We first show that COs belonging to the algebra generating the allowed transformations can be obtained from the kernel of a quantum covariance matrix. Their number determines the dimension of the accessible orbit, while their algebraic structure and eigenvalues provide further constraints on state convertibility. We establish general bounds on the number of COs and prove that highest-weight states, like  separable states in the case of $n$ qudits, attain the maximum. We also show how these operators generate parent Hamiltonians having the characterized state as an exact ground state. We then also consider positive zero-eigenvalue COs, which encode forbidden directions of reduced density matrices and incorporate rank-based constraints into the same framework. Applying this approach to multiqubit SLOCC, we obtain complementary classifications from algebraic and rank COs. While for $n\leq 3$ qubits both lead to the usual classes,  for $n=4$ qubits they distinguish 17 and 20 types separately, and 29 when combined. The combined classification yields 287 fixed types and $4$ continuous families for $n=5$  qubits, and 4926 fixed types and 147 continuous families for $n=6$ qubits, providing, to our knowledge, the most refined exhaustive classification of SLOCC-convertibility classes based on such data.

\end{abstract}

\maketitle

\section{Introduction}

One of the basic questions at the heart of quantum resource theories (QRTs) is whether a given set of allowed ``free'' operations can transform a quantum state into another~\cite{chitambar2019quantum}. For pure states, invertible transformations, typically obtained as some branch of a specific quantum instrument~\cite{chitambar2014everything}, partition the Hilbert space into orbits of interconvertible states. For instance, in entanglement theory, these transformations give stochastic local operations and classical communication (SLOCC)~\cite{dur2000three,miyake2003classification}, while similar nonunitary operations can be obtained for more general Lie-algebra-based QRTs by complexifying their free groups~\cite{barnum2003generalizations,diaz2025unified}. However, understanding the orbits associated with these transforms can be difficult even for relatively small systems. Indeed,  the Schmidt decomposition completely solves the problem for bipartite entanglement, and in three qubits one recovers the familiar product, biseparable, W, and GHZ classes. Already for four qubits, continuous parameters appear and the orbit structure becomes considerably richer~\cite{verstraete2002four,dietrich2022classification,jia2015classification}.

\begin{figure*}[t]
\centering
\includegraphics[width=1\linewidth]{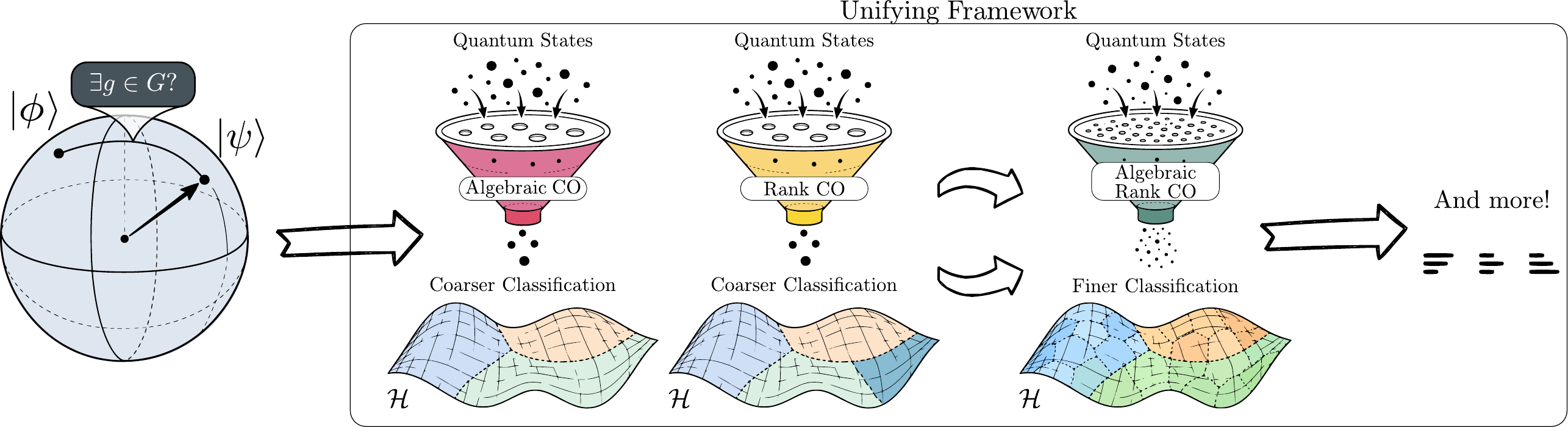}
\caption{\textbf{Conserved-operator classification of quantum states.}
Given two states $\ket\psi$ and $\ket\phi$, we ask whether an allowed free transformation $g\in G$ can connect them. In this work we focus on two types of COs, dubbed algebraic and rank, but our framework is more general than simply those two. Algebraic COs provide a classification based on the symmetries of the states, while rank COs add information from forbidden directions in reduced density matrices. Combining the two refines the resulting classes. States with inequivalent CO data cannot belong to the same orbit, providing a hierarchy of constraints on state interconvertibility within a common framework. }
\label{fig:schematic}
\end{figure*}

Over the past few decades several approaches have been developed to classify state interconvertibility orbits under free operations. These include normal forms for representative states~\cite{verstraete2002four,chterental2006normal}, polynomial invariants and covariants~\cite{luque2003polynomial,gour2013classification,walter2013entanglement,gharahi2020fine,briand2003complete,luque2006algebraic,li2012slocc}, rank-based techniques~\cite{dur2000three,li2012method,li2012classification,chen2006range},  and methods based on the stabilizer of a state~\cite{carteret2000local,walck2007maximum,gour2017almost,slowik2020link,dietrich2022classification}. While all of these approaches constrain the same orbit structure, they are generally treated separately and can lead to different classifications of the same set of states. It is therefore not always clear how the information captured by one approach is related to that obtained from another, or how these different descriptions can be combined. Stabilizer-based methods are particularly appealing from a symmetry and physics-informed perspective, as any transformation connecting two states must also relate the conserved operators that leave them invariant. In many SLOCC treatments, however, the stabilizer is defined at the vector level, and therefore retains only generators that annihilate the state. This can miss projective symmetries for which the state is an eigenvector with a nonzero eigenvalue. More generally, even the projective stabilizer is restricted to operators in the algebra generating the free transformations, and therefore does not capture conserved operators outside this subspace.

In this work, we take a broader conserved-operator perspective. We call an operator $Q$ a conserved operator (CO) of a pure state $\ket{\psi}$ whenever $\ket{\psi}$ is an eigenstate of $Q$. When $Q$ belongs to the Lie algebra generating the Lie group of free transformations, the COs are precisely the infinitesimal operations that leave the state unchanged as a ray, i.e., the so-called projective stabilizers. Importantly, our definition of a CO is more general. When the eigenvalue is zero, additional conserved operators are possible, providing information that is not contained in the projective stabilizer. More generally, we show that several of the standard approaches to SLOCC classification can be understood as different ways of extracting CO information, meaning that our work provides a unifying umbrella for a wide range of state characterization approaches.

This point of view also gives a direct connection with measurable correlations. The COs in the generator algebra are obtained from the null space of a quantum covariance matrix~\cite{qi2019determining,chertkov2018computational,petrovich2024covariance,petrovich2025conserved}. Their dimension determines the dimension of the accessible orbit, while their algebraic structure gives further necessary conditions for state convertibility. We show that the largest CO algebras correspond to highest-weight states, recovering product states for local qubit transformations and more generally generalized coherent states~\cite{barnum2003generalizations,barnum2004subsystem,perelomov1977generalized,gilmore1974properties,zhang1990coherent}. Importantly, all the information used in our classification comes from one- and two-body correlations that can be efficiently measured on states prepared in quantum computers, without requiring full state tomography.

To exemplify our framework, we apply it to pure qubit states under SLOCC. Both the algebraic and positive two-body zero-eigenvalue COs, dubbed rank COs, individually recover the usual SLOCC classes through $3$ qubits and provide a symmetry-based classification for larger systems. For $4$ or more qubits, aggregating both types of COs substantially refines the classification beyond what each type can detect on its own (see Fig.~\ref{fig:schematic}).  Our framework allows us to naturally combine stabilizer-based and rank-based classifications, which are usually treated separately. Finally, we show that the same COs can be used to construct parent Hamiltonians for which the characterized state is an exact eigenstate or ground state.

\section{Results}

\subsection{Conserved operators as zero fluctuations}

Consider a finite-dimensional Hilbert space $\mathcal H$ of dimension $D_{\mathcal H}$, with associated space of linear operators denoted as ${\rm End}(\mathcal{H})$. From here, we can already introduce the notion of COs:
\begin{definition}[Conserved operators (COs)]\label{def:CO}
    Let $\ket{\psi}$ be a normalized pure state in $\mathcal{H}$. We say that $Q\in{\rm End}(\mathcal{H})$ is a CO of $\ket{\psi}$ if $Q\ket{\psi}=\lambda\ket{\psi}$ for some $\lambda\in\mathbb{C}$.
\end{definition}
\noindent The special case $\lambda=0$ recovers the zero-eigenvalue vector stabilizer condition considered in~\cite{gour2017almost,dietrich2022classification}. However, we first focus on the generic eigenvalue case, and return to the additional structure that arises for $\lambda=0$ later.

Given that throughout this work we assume that the free transformations form a Lie group $G$ whose complex reductive Lie algebra $\mathfrak g\subseteq\operatorname{End}(\mathcal H)$ acts irreducibly on $\mathcal H$, the most natural place to look for COs is within $\mathfrak g$ itself. Indeed, if $Q\in\mathfrak g$ is a CO of $\ket\psi$,  then $e^{tQ}\ket\psi=e^{t\lambda}\ket\psi$, so the corresponding one-parameter transformation leaves the ray $[\psi]\in\mathbb{P}(\mathcal{H})$ invariant\footnote{Conversely, if $e^{tQ}\ket\psi=e^{t\lambda}\ket\psi$ for all $t$, then $Q$ is a CO.}. Here, $\mathbb{P}(\mathcal{H})$ denotes the projective Hilbert space, defined as the set of rays in $\mathcal H$, i.e., normalized states identified up to multiplication by a nonzero complex scalar. Hence, the COs contained in $\mathfrak g$ are precisely the infinitesimal generators of the projective stabilizer of the state. We denote the complete space of algebraic COs as
\begin{equation}
 \pstab[\psi]=
 \{Q\in\mathfrak g\mid Q\ket\psi=\chi_\psi(Q)\ket\psi\}.
 \label{eq:co}
\end{equation}
The eigenvalue $\chi_\psi(Q)$ is a linear function of $Q$. Moreover, for any two COs $Q,Q'\in\pstab[\psi]$, their commutator annihilates $\ket\psi$. The algebraic COs therefore form a Lie subalgebra of $\mathfrak g$, with $\chi_\psi$ being their character. 

Notably, we can fully characterize the algebraic COs of a given state without knowledge of all the amplitudes of the state vector. Given a complex basis $O^1,\ldots,O^{D_{\mathfrak g}}$ of
$\mathfrak g$, with $D_{\mathfrak g}$ the dimension of ${\mathfrak g}$, we can subtract each expectation value to define the operators
$\Delta O^\mu=O^\mu-\langle O^\mu\rangle_\psi\id$.
The associated covariance matrix
\begin{equation}
 {C}^{\mu\nu}(|\psi\rangle)
 =\langle(\Delta O^\mu)^\dagger\Delta O^\nu\rangle_\psi\in\mathbb{C}^{D_{\mathfrak g}\times D_{\mathfrak g}}\,,
 \label{eq:covariance}
\end{equation}
is positive semidefinite.
This gives the following criterion~\cite{qi2019determining,petrovich2024covariance,petrovich2025conserved}.

\begin{proposition}[Covariance criterion]
\label{prop:covariance}
An operator $Q=\sum_{\mu=1}^{D_{\mathfrak g}} u_\mu O^\mu$ is an algebraic CO of $\ket\psi$  as per Definition~\ref{def:CO}
if and only if its coefficient vector $\bm{u}\in\mathbb{C}^{D_{\mathfrak g}} $ belongs to $\ker \bf{C}(|\psi\rangle)$.
\end{proposition}
\noindent To see both implications, note that $ \bm{u}^\dagger {\bf C} \bm{u} =\| (Q-\langle Q\rangle_\psi\id)\ket\psi\|^2$. For a positive matrix the left-hand side vanishes exactly on its kernel. The right-hand side vanishes exactly when $Q$ preserves the state ray. The basis of $\mathfrak g$ thus identifies the
complex vector spaces $\ker {\bf C}$ and $\pstab [\psi]$. 

Next, note that because the allowed transformations are generally nonunitary, the CO algebra is not restricted to Hermitian operators. This distinction can be seen directly from the covariance matrix. In particular, let us write ${\bf C}={\bf C}_++{\bf C}_-$ with
\begin{equation}
C_\pm^{\mu\nu} =\frac12\left\langle
(\Delta O^\mu)^\dagger\Delta O^\nu
\mathbin{\pm}\Delta O^\nu(\Delta O^\mu)^\dagger \right\rangle_\psi .
\label{eq:covariance-split}
\end{equation}
The symmetrized part $\bf C_+$ is positive semidefinite and satisfies
\begin{equation}
\bm u^\dagger {\bf C}_+\bm u = \frac12\left(
 \|\Delta Q\ket\psi\|^2+ \|(\Delta Q)^\dagger\ket\psi\|^2
 \right),
 \label{eq:symmetric-fluctuation}
\end{equation}
for $Q=\sum_\mu u_\mu O^\mu$. It follows that
$\bm u\in\ker {\bf C}_+$ if and only if both $Q$ and $Q^\dagger$ are COs of $\ket\psi$. In particular,
$\ker {\bf C}_+\subseteq\ker{\bf C}_-$. When the basis operators $O^\mu$ are Hermitian, one has ${\bf C}_+=\operatorname{Re}{\bf C}$ and
$-\mathrm{i}{\bf C}_-=\operatorname{Im}{\bf C}$. Hence, the real vectors in $\ker{\bf C}_+$ determine the Hermitian COs, while the full kernel of ${\bf C}$ is needed to identify general non-Hermitian COs. 

\subsection{CO algebra and orbit of accessible states}
Having determined the algebraic COs of a given state, we can use them to characterize the orbit $G\cdot[\psi]$ of states reachable under the allowed transformations in $G$. As schematically shown in Fig.~\ref{fig:orbits}, each operator $O\in\mathfrak g$ generates an infinitesimal change $O\ket\psi$. Since the component parallel to $\ket\psi$ does not change the ray, the corresponding tangent vector is~\cite{provost1980riemannian,wierichs2023symmetric}
\begin{equation}
\mathcal{T}_{[\psi]}(O)
 =(\id-\ket\psi\bra\psi)O\ket\psi
 =\Delta O\ket\psi.
 \label{eq:tangent}
\end{equation}
As the $O$ spans $\mathfrak g$, these vectors span the tangent space to the orbit at $[\psi]$, and an algebraic CO is precisely an operator for which $\mathcal{T}_{[\psi]}(Q)=0$. Moreover, since the entries of the covariance matrix $\bm C$ are the inner products between these tangent vectors, $\bm C$ is their Gram matrix. It follows that
\begin{equation}
 {\rm rank}\ {\bf C}
 =\dim G\cdot[\psi]
 = D_{\mathfrak g}- D_{\pstab [\psi]},
 \label{eq:geometry}
\end{equation}
where $D_{\mathfrak p} = \dim \pstab [\psi]$. Equation \eqref{eq:geometry} shows that the dimension of the CO algebra directly determines the dimension of the accessible orbit, and as such provides the first necessary condition for state interconvertibility: $\ket\phi$ can potentially be transformed into $\ket\psi$ only if $D_{\pstab[\psi]}=D_{\pstab [\phi]}$. 

Moreover, we can further push the necessary condition by exploiting the full structure of the projective algebras. If $Q|\psi\rangle=\chi_\psi(Q)|\psi\rangle$ and $\ket{\phi}=g\ket{\psi}$, then $(g Q g^{-1})|\phi\rangle=\chi_\psi(Q)|\phi\rangle$, implying the following proposition.

\begin{proposition}[CO algebra equivalence under orbit]\label{prop:co-equiv}
Let $\ket{\phi}=g\ket{\psi}$ with $g\in G$. Then
\begin{equation}
\pstab [\phi] = {\rm Ad}_g \pstab [\psi],\qquad \chi_\phi = \chi_\psi \circ {\rm Ad}_{g^{-1}},
 \label{eq:conjugacy}
\end{equation}
where ${\rm Ad}_g(O)=g O g^{-1}$ and $\chi \!\circ\! {\rm Ad}_{g}(O) \!=\! \chi({\rm Ad}_{g}O)$.
\end{proposition}

Hence, states belonging to the same orbit must have CO algebras related by a similarity transformation induced by some $g\in G$. In particular, the similarity map $Q\mapsto gQg^{-1}$ preserves the commutation relations and eigenvalue characters of the COs. States whose CO algebras cannot be related by such an allowed similarity transformation therefore cannot be interconverted by an element of $G$. Figure~\ref{fig:orbits} illustrates this relation.

\begin{figure}[!t]
\centering
\includegraphics[width=0.93\linewidth]{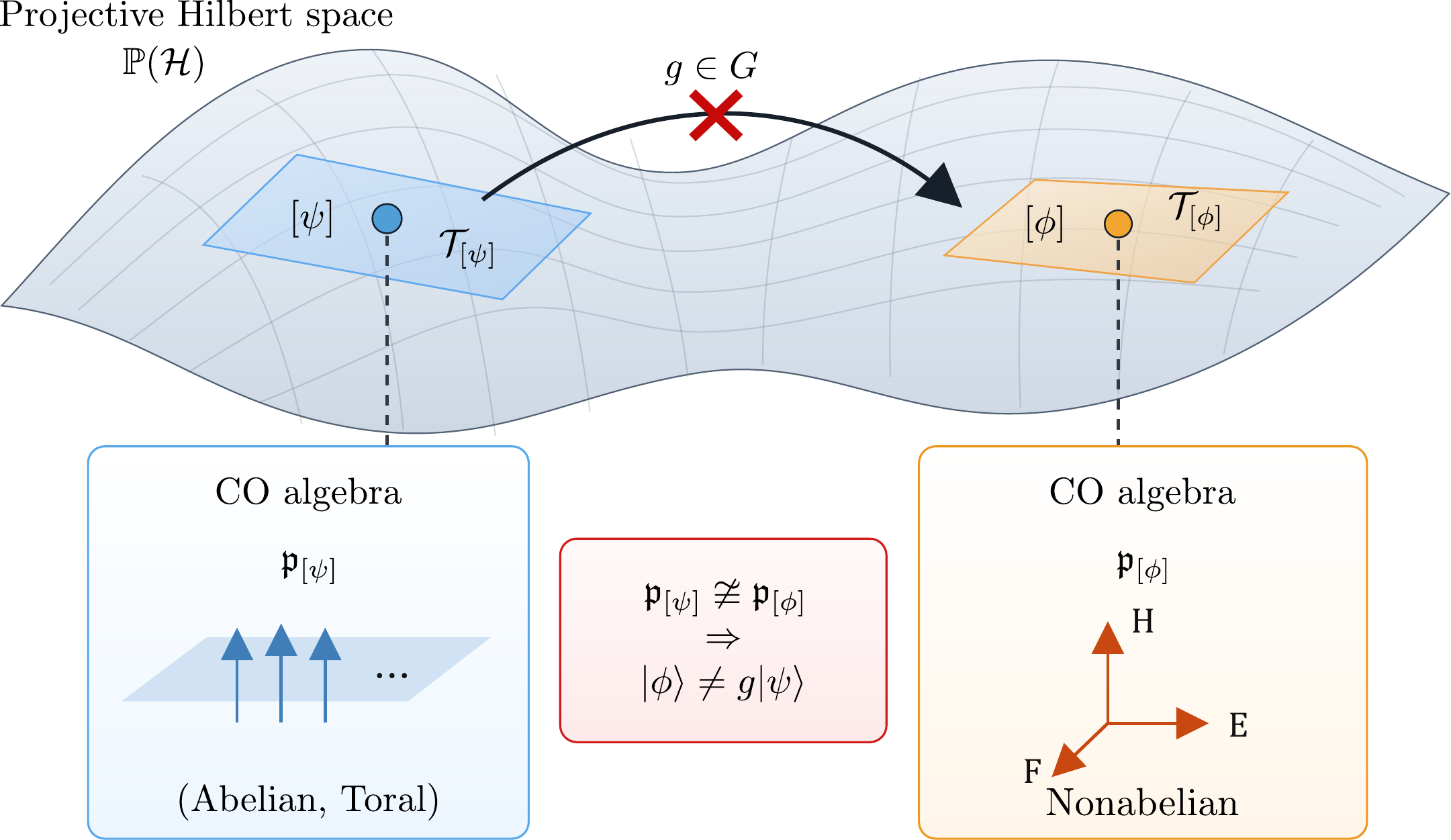}
\caption{\textbf{Algebraic  COs  and the tangent planes of state orbits.} The grey curved surface represents projective Hilbert space. The colored planes schematically represent the orbit tangents at the indicated states under the allowed transformations. After removing the component along the state ray, the vectors $\Delta O^\mu\ket\psi$ span each tangent plane. Their Gram matrix is the covariance $\bm C$, and a CO is a linear combination producing the zero tangent vector. A map between two states conjugates their CO algebras. The two illustrated algebras have different commutation relations, which rules out conversion even if their dimensions agree.}
\label{fig:orbits}
\end{figure}

\subsection{Structure of the CO algebra}

As previously mentioned, characterizing the algebra $\pstab[\psi]$ is extremely important to understand the states that $\ket \psi$ can be converted into. For such a characterization, we make use of the Cartan--root decomposition of $\mathfrak g$~\cite{hall2013lie}. For the  complex reductive Lie algebra $\mathfrak g$, let $\mathfrak g^z$ denote a Cartan subalgebra and let $\mathcal R$ be the corresponding set of roots. After choosing a set of positive roots $\mathcal R_+$, one can always decompose the algebra as
\begin{equation}
\mathfrak g = \mathfrak g^z \oplus \mathfrak g^+ \oplus \mathfrak g^-, \qquad \mathfrak g^\pm = \bigoplus_{\gamma\in\mathcal R_+} \mathfrak g^{\pm\gamma}.
\label{eq:cartan-root}
\end{equation}
Here we recall that the operators in $\mathfrak g^z$ are mutually commuting and semisimple, while the root spaces $\mathfrak g^{\pm\gamma}$ contain raising and lowering operators satisfying
\begin{equation}
[O^z,O^{\pm\gamma}] = \pm\gamma(O^z)O^{\pm\gamma}, \qquad O^z\in\mathfrak g^z.
\end{equation}
In what follows we will denote as $\ket {\psi_{\mathrm{hw}}}$ the highest-weight state of the algebra, i.e., a simultaneous eigenstate of the Cartan generators that is annihilated by all positive-root generators $O^{+\gamma}$.

The same decomposition can be used to organize the CO algebra of a state. In general, since $\pstab[\psi]$ need not be aligned with the reference decomposition in Eq.~\eqref{eq:cartan-root}, we can choose an adapted decomposition
$\mathfrak g=\mathfrak g^{ z'}\oplus\mathfrak g^{ +'}\oplus\mathfrak g^{-'}$
such that
\begin{equation}
\pstab[\psi] = \pstab^{z'}\![\psi] \oplus \pstab^{+'}[\psi] \oplus \pstab^{-'}[\psi], \quad \pstab^{\mu'}[\psi] \subseteq \mathfrak g^{\mu'}.
\label{eq:co-root-decomposition}
\end{equation}
Here, the prime index denotes that there exists a transformation $g\in G$ that maps the adapted decomposition to the reference one. Now, the Cartan sector $\pstab^{z'}[\psi]$ contains the semisimple conserved directions, while $\pstab^{\pm'}[\psi]$ contain the conserved root directions. Since root operators are nilpotent, any CO in $\pstab^{\pm'}[\psi]$ necessarily has zero eigenvalue, implying that the character $\chi_\psi$ vanishes on these sectors, and any nonzero CO eigenvalue is carried by the Cartan component. 

\subsection{Bounds on conserved operators and orbit dimensions}

Having characterized the algebra generated by the algebraic  COs, we now shift gears and ask the questions: \textit{What is the minimum number of algebraic  COs a state can possess?} and inversely, \textit{What are the states which possess the largest number of algebraic  COs, or equivalently, the smallest orbits?} For an irreducible representation of a complex reductive group $G$, the answer to these questions is captured in the following theorem, whose proof can be found in the SI.

\begin{theorem}[Bounds on algebraic  COs and orbit dimensions]
\label{thm:highest}
For every non-zero $\ket{\psi}\in\mathcal H$ we find
\begin{equation}
 \max\{0,D_{\mathfrak g}+1-D_{\mathcal H}\} \leq D_{\pstab[\psi]} \leq
 D_{\pstab[\psi_{\rm hw}]}\,.
\end{equation}
The upper bound on the number of algebraic  COs, and equivalently the lower bound on the orbit dimension, is saturated if and only if
$[\psi]\in G\cdot[\psi_{\mathrm{hw}}]$.
\end{theorem}

The lower bound immediately implies that whenever $D_{\mathfrak g}+1>D_{\mathcal H}$, every state necessarily possesses nontrivial COs. At the opposite extreme, Theorem~\ref{thm:highest} shows that highest-weight states possess the maximum number of COs. From Proposition~\ref{prop:co-equiv}, every state in their orbit under $G$, i.e., the so-called generalized coherent states~\cite{perelomov1977generalized,gilmore1974properties,zhang1990coherent}, therefore has the same maximal number of COs.

\subsection{Connection to parent Hamiltonians}

Here we briefly discuss the fact that the algebraic COs can also be used to engineer Hamiltonians for which the characterized state is an exact eigenstate~\cite{perez2006matrix,cerezo2016factorization,chertkov2018computational,petrovich2024covariance,petrovich2025conserved,shiraishi2017systematic}. Let us denote as $Q^\alpha$ and $Q_h^{\alpha}$ general and Hermitian COs (obtained from $\ker \bf C_+$), respectively. Then, for real coefficients $h_\alpha$ and a Hermitian matrix $\bf V$, the Hamiltonian
\begin{equation}
 H=\sum_\alpha h_\alpha Q_h^\alpha +\sum_{\alpha,\beta}V_{\alpha\beta}(\Delta Q^\alpha)^\dagger\Delta Q^\beta\,,
 \label{eq:parent-eigenstate}
\end{equation}
has $\ket\psi$ as an exact eigenstate, with
\begin{equation}
 H\ket\psi = \left(\sum_\alpha h_\alpha \lambda^\alpha \right)\ket\psi.
\end{equation}
The first term fixes the corresponding eigenvalue, while the quadratic part annihilates $\ket\psi$ and can therefore be varied without changing the eigenstate. A particularly useful case is obtained by taking $h_\alpha=0$ and ${\bf V}>0$. The resulting positive Hamiltonian
\begin{equation}
 H_0=\sum_{\alpha,\beta}V_{\alpha\beta}(\Delta Q^\alpha)^\dagger\Delta Q^\beta
 \label{eq:parent}
\end{equation}
has ground space $\ker H_0=\bigcap_\alpha\ker\Delta Q^\alpha$, meaning that $\ket\psi$ is the unique ground state whenever the common kernel of the centered COs is one-dimensional.

As we can see, this connects the CO construction with covariance-based inverse Hamiltonian methods~\cite{chertkov2018computational,petrovich2025conserved}. Different positive choices of $\bf V$ lead to different Hamiltonians with the same ground space, while Hermitian CO terms can be added without changing the target eigenstate. 

\subsection{Positive zero-eigenvalue COs}

Up to this point we have focused on algebraic COs belonging to the algebra $\mathfrak g$ associated to the free transformations, whose structure fixes the similarity transformation $Q\mapsto gQg^{-1}$. Definition~\ref{def:CO}, however, is more general, and the case of zero eigenvalue allows additional transformation laws.

Let $P\geq0$ be a CO satisfying
\begin{equation}\label{eq:positive-CO}
    P\ket{\psi}=0.
\end{equation}
Clearly,  this condition is equivalent to $    \bra{\psi}P\ket{\psi}=0$, meaning that after an overall normalization, $P$ can be interpreted as a measurement effect whose outcome occurs with zero probability. In this sense, positive zero-eigenvalue COs describe configurations that are forbidden by the state. Interestingly, the transformation law for these operators follows naturally from this interpretation. Writing $P=M^\dagger M$, Eq.~\eqref{eq:positive-CO} is satisfied if and only if $M\ket{\psi}=0$. Hence, for any state in the orbit, $\ket{\phi}=g\ket{\psi}$, the transformed constraint $M g^{-1}$ annihilates $\ket{\phi}$. The transformed operator $\widetilde P=(g^{-1})^\dagger P g^{-1}$ is positive and satisfies $\ker\widetilde P=g\,\ker P$. Hence, positive null directions are transported consistently along the same state orbit, but according to congruence rather than similarity. Moreover, this congruence transformation preserves positivity even when $g$ is nonunitary.

As we can see, positive COs can therefore provide additional necessary conditions for state convertibility, including in cases where the projective stabilizer in $\mathfrak g$ is trivial. Of course, if we allow generic positive COs, these constraints would contain essentially complete information about the state. We therefore restrict below to few-body positive COs, which provide a finer classification while directly connecting several SLOCC classification methods within the same framework.

\subsection{Application to multipartite entanglement}

We now instantiate our framework to states of $n$ distinguishable subsystems, with local dimensions $d_1,\ldots,d_n$. A general pure state can be written as
\begin{equation}
 \ket\psi=\sum_{i_1,\ldots,i_n} \Gamma_{i_1\cdots i_n} \ket{i_1}\otimes\cdots\otimes\ket{i_n},
 \quad \mathcal H=\bigotimes_{p=1}^n\C^{d_p}.
 \label{eq:distinguishable-state}
\end{equation}

\subsubsection{One-body COs}

Here, the free transformations are invertible local maps
$g_1\otimes\cdots\otimes g_n\in G^{\mathrm{dis}}=\bigotimes_{q=1}^n\mathbb{SL}(\C^{d_q})$, corresponding to SLOCC transformations of $n$ qudits with local dimensions $d_1,\ldots,d_n$~\cite{miyake2003classification}. These are generated by the local algebra
\begin{equation}
\mathfrak g^{\mathrm{dis}} =\bigoplus_{p=1}^n\slc_{d_p}(\C),
\label{eq:distinguishable-algebra}
\end{equation}
with $\mathfrak{sl}_{d_p}(\mathbb C) =
\operatorname{span}_{\mathbb C}
\left\{ G_p^{ij}-\frac{\delta^{ij}}{d_p}\id_p \right\}$ and $G_p^{ij}=|i_p\rangle\langle j_p|$.

For this local algebra, Theorem 1 immediately yields a general bound on the dimension of the algebraic CO space. Defining $\bm d=(d_1,\ldots,d_n)$ and $d=\sum_{p=1}^n d_p$, we obtain
\begin{equation}
\max\!\left\{0,\,\|\bm d\|^2+1-\prod_{p=1}^n d_p \right\} \leq D_{\pstab [\psi]} \leq \|\bm d\|^2-d,
\label{eq:quditbound}
\end{equation}
where $\|\bm d\|^2=\sum_{p=1}^n d_p^2$. The upper bound is saturated by product states, and only by product states.
Indeed, choosing a local basis such that
$|\psi_{\rm prod}\rangle = \otimes_{p=1}^n |0_p\rangle$, the corresponding algebraic COs are spanned, at each site $p$, by $G_p^{0\bar j}$ and $G_p^{\bar{i} \bar{j}}$, where here and in what follows, unbarred (barred) indices denote levels within (orthogonal to) the support of the state, and diagonal generators are understood in their traceless form. Thus, the local projective stabilizer at site $p$ has dimension $d_p(d_p-1)$, so that $D_{\pstab[\psi_{\rm prod}]} = \|\bm d\|^2-d$. 

For $n=1$, there is a single SLOCC family, represented by $|0\rangle$,
in agreement with Eq. \eqref{eq:quditbound}, since the upper and lower bounds coincide at $d (d-1)$.

For $n=2$, the orbit structure is completely determined by the Schmidt
rank. Since all nonzero Schmidt coefficients can be absorbed by invertible local transformations, every state of rank $r$ is SLOCC-equivalent to $|{\rm GHZ}_2^{\,r}\rangle
= \frac{1}{\sqrt r} \sum_{k=0}^{r-1}|k_1\rangle|k_2\rangle$ with $1\leq r\leq\min(d_1,d_2)$, implying that there are $\min(d_1,d_2)$ bipartite SLOCC orbits, labeled by $r$. In the Schmidt basis, the corresponding algebraic COs are $G_1^{ij}-G_2^{ji}$, $G_p^{i\bar j}$, $G_p^{\bar i\bar j}$ ($p=1,2$). Thus, the CO dimension distinguishes all these
orbits and decreases monotonically with the Schmidt rank:
\begin{equation}
D_{\pstab[{\rm GHZ}_2^{r}]} = \|\bm d\|^2-r(d-r)-1.
\end{equation}

More generally, Eq. \eqref{eq:quditbound} implies that every multipartite state necessarily possesses nontrivial algebraic COs whenever $\|\bm d\|^2+1>\prod_{p=1}^n d_p$. In particular, for $n=3$ and local dimensions $(2,2,2)$, $(2,2,3)$,
and $(2,3,3)$, the lower bound in Eq. \eqref{eq:quditbound} yields
$2$, $3$, and $2$ COs, respectively. In contrast, for $n=3$ with $d_p\geq 3 \ \forall \ p$, or for $n=4$ with $d_p \geq 2 \ \forall \ p$, the lower bound becomes nonpositive, implying that a generic state can have no algebraic COs in $\mathfrak g^{\rm dis}$.

\subsubsection{Rank COs}

So far we have just considered COs  which are linear combinations of local one-body operators, hence closed under commutation. For $m \geq2$, $m$-body operators are in general no longer closed  under such operation. However, we can  turn our attention to the set of local COs, i.e., operators of the form $Q=Q_S \otimes \id_{\overline S}$, acting non-trivially only on a small subset $S$ of qudits, with $\overline S$ its complement,  which do form a closed commutation algebra. For $|S|=1$, these operators are already contained in the previous one-body algebra \eqref{eq:distinguishable-algebra}, so we will focus on subsets with $|S|\geq 2$, for which local
COs can provide additional information.

For $\mathfrak g_S=\mathfrak{sl}(\mathcal H_S)$ the algebra of traceless operators supported on $S$, we define the corresponding local CO space as
\begin{equation}
\pstab_S[\psi] = \left\{
Q_S\in\mathfrak g_S: (Q_S\otimes\id_{\bar S})|\psi\rangle = \chi_\psi(Q_S)|\psi\rangle
\right\}.  
\end{equation}
Notice that $Q_S\in\mathfrak p_S[\psi]$ if and only if $\left\|(\Delta Q_S\otimes I_{\bar S})
|\psi\rangle\right\|^2 \!=\! \operatorname{Tr}_S[ \Delta Q_S^\dagger \Delta Q_S\rho_S]\!=\! \|\Delta Q_S\sqrt{\rho_S}\|^2\!=\!0$, where $\Delta Q_S = Q_S - \chi_\psi(Q_S) \id_S$ and $\rho_S=\operatorname{Tr}_{\bar S}|\psi\rangle\langle\psi|$. This immediately yields the following criterion \footnote{Notice that
$\|\Delta Q_S\sqrt{\rho_S}\|_{\rm HS}^2$ vanishes if and only if $\Delta Q_S\rho_S=0$, since $\rho_S$ and $\sqrt{\rho_S}$ have the same support.}:
\begin{proposition}[Reduced state kernel criterion]
\label{prop:rdm-kernel}
A local operator $Q=Q_S \otimes \id_{\overline S}$ is a CO of $\ket\psi$ with eigenvalue $\lambda$ if and only if
\begin{equation}
(Q_S - \lambda \id_S) \rho_S = 0.
\end{equation}
\end{proposition}

Then, for $G_S^{ij} = |i_S\rangle \langle j_S|$, the subspace of local COs is spanned by the operators $G_S^{i \bar j}$, $G_S^{\bar i \bar j}$, with $\ket{\bar j_S} \in \mathcal \ker\, \rho_S$, implying, in particular,
\begin{equation}
{\rm rank}\ \rho_S =
D_{\mathcal H_S}- \frac{D_{\pstab_S[\psi]}}{D_{\mathcal H_S}}.
\end{equation}
Moreover, we can state the following proposition:
\begin{proposition}[Local CO and reduced-kernel equivalence under SLOCC] \label{prop:pair} Let $g = \bigotimes_{p} g_p$ and $g_S = \bigotimes_{p\in S}g_p$. Then, the SLOCC equivalence $\ket{\phi} = g \ket{\psi}$ implies that
$(\pstab_S[\phi],(\chi_{S})_\phi)$ and $(\pstab_S[\psi],(\chi_{S})_\psi)$ are related through \eqref{eq:conjugacy} for $g \rightarrow g_S$, or equivalently, 
\begin{equation}
\ker\rho_S[\phi] =
(g_S^\dagger)^{-1}\ker\rho_S[\psi].
\label{nullcon}
\end{equation}
\end{proposition}

The transformation of the local CO algebra and its character follows directly from Proposition 2, while the invertibility of $g_{\bar S}$ ensures $\operatorname{supp}\rho_S[\phi]
=g_S\operatorname{supp}\rho_S[\psi]$,
whose orthogonal complement yields Eq. \eqref{nullcon}. 

Since local COs are completely determined by the kernel of the
corresponding reduced density matrix, in what follows, we will
characterize the local CO data through $\ker\rho_S$ itself. Equivalently,
we may restrict attention to the positive zero-eigenvalue COs
\begin{equation}
P_S^{\bar k} = \left(|\bar k_S\rangle\langle \bar k_S|\right)
\otimes \mathbb I_{\bar S},
\qquad |\bar k_S\rangle\in\ker\rho_S, 
\end{equation}
and we will refer to these as rank COs. This construction directly recovers the reduced state rank constraints commonly used in SLOCC classification \cite{dur2000three,li2012method,li2012classification}. Importantly, we note that following  Eq.~\eqref{nullcon}, given some group element $g_S$, then  $\left(|\bar k_S\rangle\langle \bar k_S|\right)$ transforms as $(g_S^{-1})^\dagger\left(|\bar k_S\rangle\langle \bar k_S|\right) \,g_S^{-1}$, i.e., via congruence and not via similarity as the algebraic COs in Proposition~\ref{prop:co-equiv}. 

In the remainder of this work, we restrict the rank-CO data to
two-site ($p\neq q$) subsets, $|S|=2$. 
This is the simplest case beyond the one-body operators already contained 
in $\mathfrak g^{\rm dis}$. Moreover, the corresponding reduced state $\rho_{pq}$ is completely determined by one- and two-body correlations, i.e., by the same level of correlation data required to construct the covariance matrix of the algebraic COs. 

\subsubsection{Qubit systems}

We now focus on qubit systems, $d_p=2 \ \forall \ p$, such that $\mathfrak g^{\rm dis}$ is the algebra generated by the three Pauli operators on each site $p$. Both the covariance matrix and the reduced density matrices $\rho_{pq}$ for $p<q$ are obtained entirely from one- and two-body Pauli correlations:
\begin{subequations}
\begin{eqnarray*}
C^{p \mu,q\nu} \!&=&\!\langle\sigma_p^{\mu} \sigma_q^\nu\rangle - \langle\sigma_p^{\mu}\rangle \langle \sigma_q^\nu \rangle \\
\rho_{pq} &=&
\rho_p\otimes\rho_q
+\tfrac{1}{4}\sum_{\mu,\nu}
C^{p\mu,q\nu}\,
\sigma^\mu \otimes \sigma^\nu,
\end{eqnarray*}
\end{subequations}
where $\rho_p=\frac{1}{2}(
\id + \sum_{\mu} \langle\sigma_p^\mu\rangle\,\sigma^\mu)$.

We ask how far the algebraic and rank CO data introduced above take us in classifying multipartite entanglement. Rather than retaining only the dimension of the algebraic stabilizer or the ranks of the reduced states, we keep the complete algebraic CO structure together with its eigenvalue character and the full pair-kernel data. These complementary constraints lead to a refined classification of multiqubit states. 

In the following, we employ the method and algorithm developed in Section B of the Methods to obtain an exhaustive classification of all allowed and realizable algebraic and rank CO data. We first determine all possible algebraic CO data, then incorporate the compatible pair rank CO spaces, identifying equivalent configurations under simultaneous invertible local transformations and qubit permutations. Finally, we verify that every resulting type is realized by a pure state with exactly the prescribed CO data.
    
\begin{theorem}[Classification from algebraic COs]
\label{thm:stabilizer-classification}
Under invertible local maps and qubit permutations, the algebraic CO data
$(\pstab_[\psi],\chi_\psi)$
give $1,2,4,17$ types for $1$ through $4$ qubits, $76$ fixed types and $4$ continuous families for $5$ qubits, and $430$ fixed types and $32$ continuous families for $6$ qubits, as shown in Table \ref{tab:counts}.  Every fixed type and every allowed parameter value is realized by a pure state.
\end{theorem}

\begin{table}[tbp]
\centering
\footnotesize
\setlength{\tabcolsep}{1.5pt}
\renewcommand{\arraystretch}{0}
\setlength{\arrayrulewidth}{0.4pt}

\ifdefined\CountsUnit\else
  \newlength{\CountsUnit}
\fi
\setlength{\CountsUnit}{\dimexpr(\linewidth-32.2pt)/6\relax}

\providecommand{\CountsCell}[1]{%
  \parbox[c][13.6pt][c]{\linewidth}{\centering #1}%
}
\providecommand{\CountsSubhead}[1]{%
  \parbox[c][22pt][c]{\linewidth}{%
    \centering\scriptsize\bfseries #1}%
}
\providecommand{\CountsHead}[1]{%
  \parbox[c][25pt][c]{\dimexpr2\CountsUnit+3.4pt\relax}{%
    \centering\footnotesize\bfseries #1}%
}

\begin{tabular}{
|>{\columncolor{cyan!10}[\tabcolsep][\tabcolsep]
    \centering\arraybackslash}m{8pt}
|>{\columncolor{red!10}[\tabcolsep][\tabcolsep]
    \centering\arraybackslash}m{.65\CountsUnit}
|>{\columncolor{red!10}[\tabcolsep][\tabcolsep]
    \centering\arraybackslash}m{1.35\CountsUnit}
|>{\columncolor{yellow!15}[\tabcolsep][\tabcolsep]
    \centering\arraybackslash}m{.65\CountsUnit}
|>{\columncolor{yellow!15}[\tabcolsep][\tabcolsep]
    \centering\arraybackslash}m{1.35\CountsUnit}
|>{\columncolor{green!10}[\tabcolsep][\tabcolsep]
    \centering\arraybackslash}m{.65\CountsUnit}
|>{\columncolor{green!10}[\tabcolsep][\tabcolsep]
    \centering\arraybackslash}m{1.35\CountsUnit}|}
\hline

\cellcolor{cyan!20}
& \multicolumn{2}{
    >{\columncolor{red!15}[\tabcolsep][\tabcolsep]}c|
  }{\CountsHead{Algebraic COs}}
& \multicolumn{2}{
    >{\columncolor{yellow!20}[\tabcolsep][\tabcolsep]}c|
  }{\CountsHead{Rank COs}}
& \multicolumn{2}{
    >{\columncolor{green!15}[\tabcolsep][\tabcolsep]}c|
  }{\CountsHead{Algebraic +\\Rank COs}} \\

\hhline{%
  >{\arrayrulecolor{black}}|%
  >{\arrayrulecolor{cyan!20}}-%
  >{\arrayrulecolor{black}}|------|%
}

\cellcolor{cyan!20}
  \raisebox{12.7pt}[0pt][0pt]{\bfseries $\bm n$}
& \cellcolor{red!15}\CountsSubhead{Fixed\\types}
& \cellcolor{red!15}\CountsSubhead{Continuous\\Families}
& \cellcolor{yellow!20}\CountsSubhead{Fixed\\types}
& \cellcolor{yellow!20}\CountsSubhead{Continuous\\Families}
& \cellcolor{green!15}\CountsSubhead{Fixed\\types}
& \cellcolor{green!15}\CountsSubhead{Continuous\\Families} \\
\hline

\CountsCell{1}
& \CountsCell{1} & \CountsCell{0}
& \CountsCell{1} & \CountsCell{0}
& \CountsCell{1} & \CountsCell{0} \\ \hline

\CountsCell{2}
& \CountsCell{2} & \CountsCell{0}
& \CountsCell{2} & \CountsCell{0}
& \CountsCell{2} & \CountsCell{0} \\ \hline

\CountsCell{3}
& \CountsCell{4} & \CountsCell{0}
& \CountsCell{4} & \CountsCell{0}
& \CountsCell{4} & \CountsCell{0} \\ \hline

\CountsCell{4}
& \CountsCell{17} & \CountsCell{0}
& \CountsCell{20} & \CountsCell{0}
& \CountsCell{29} & \CountsCell{0} \\ \hline

\CountsCell{5}
& \CountsCell{76} & \CountsCell{4}
& \CountsCell{129} & \CountsCell{1}
& \CountsCell{287} & \CountsCell{4} \\ \hline

\CountsCell{6}
& \CountsCell{430} & \CountsCell{32}
& \CountsCell{973} & \CountsCell{9}
& \CountsCell{4926} & \CountsCell{147} \\ \hline

\end{tabular}

\caption{\textbf{Classification from algebraic and rank COs.}
Algebraic COs are those in $\mathfrak g^{\mathrm{dis}}$, while rank COs
are the positive two-body zero-eigenvalue COs associated with the kernels
of reduced states. The three classifications consider algebraic COs alone,
rank COs alone, and both, respectively. Continuous families contain
infinitely many inequivalent CO data.}
\label{tab:counts}
\end{table}

For $1$ through $3$ qubits, these types coincide with the usual SLOCC classes up to permutations. In particular, the product, biseparable, W and GHZ classes for $3$ qubits have local CO dimensions $6,5,3,2$, respectively. The complete table covering from $1$  to  $4$ qubits and including representative states together with their CO generators, is presented in the Methods section, Table \ref{tabgen}.  

We now consider the classification induced by the rank CO data alone, independently of the algebraic COs. These COs admit a natural fermionic interpretation, as shown in the Methods. Encoding each qubit as a single fermion occupying one of two modes maps local qubit transformations to number-preserving one-body transformations acting within each qubit sector. In the same representation, the forbidden two-qubit directions underlying the rank COs correspond to one-body fermionic COs transferring fermions between different sectors. Thus, the algebraic and rank COs can be understood as the intra- and intersector parts of a larger one-body fermionic CO space, establishing a connection between qubit and fermionic entanglement classifications.

\begin{theorem}[Classification from rank COs]
\label{thm:rank-classification}
Under invertible local maps and qubit permutations, the rank CO data
$\{\ker\rho_{pq}[\psi]\}_{p<q}$
give $1,2,4,20$ types for $1$ through $4$ qubits,
$129$ fixed types and $1$ continuous family for $5$ qubits, and $973$ fixed types and $9$ continuous families for $6$ qubits,
as shown in Table \ref{tab:counts}.
Every fixed type and every allowed parameter value is realized by a pure state.
\end{theorem}

For $1$ through $3$ qubits, this classification again coincides with the usual SLOCC classes up to permutations. For $4$ through $6$ qubits, however, it yields a different classification from that obtained using algebraic COs alone. The complete classification for $1$ through $4$ qubits, including representative states and their CO generators, is presented in the Methods section, Table \ref{tabgen}. These differences motivate combining both types of COs, as we now discuss.

\begin{theorem}[Classification from algebraic and rank COs]
\label{thm:classification}
Under simultaneous invertible local maps and qubit permutations, the algebraic and rank CO data
$\left(
\pstab_[\psi],\chi_\psi,
\{\ker\rho_{pq}[\psi]\}_{p<q}
\right)$
give $1,2,4,29$ types for $1$ through $4$ qubits.
At $5$ qubits they give $287$ fixed types and $4$ continuous families, while at $6$ qubits they give $4926$ fixed types and $147$ continuous families, as shown in Table~\ref{tab:counts}.
Every fixed type and every allowed parameter value is realized by a pure state.
\end{theorem}

As shown in Table \ref{tab:counts}, all three classifications coincide for up to $3$ qubits. At $4$ qubits, however, the algebraic and rank COs yield $17$ and $20$ types, respectively, while their combination distinguishes $29$. This demonstrates that the two types of COs provide complementary information, as illustrated by the complete classification in Table \ref{tabgen}. In particular, the families sharing equivalent algebraic CO data are
$\{9,17\}$,
$\{11,13,20,21,27\}$,
$\{14,15,24,25,29\}$,
$\{19,26\}$, and
$\{22,23,28\}$.
Conversely, those sharing equivalent rank CO data are
$\{11,15\}$,
$\{12,14\}$,
$\{19,20,22,25\}$,
$\{21,23\}$, and
$\{26,27,28,29\}$.
All remaining families are distinct under each respective classification. Some CO types in Table \ref{tabgen} contain several SLOCC branches, for which multiple representative state families are included. This ensures that every state belonging to a given CO type is SLOCC-equivalent, up to qubit permutations, to at least one of the listed representatives for some allowed choice of parameters. These branches are nevertheless counted as a single CO type.

At $5$ qubits, the algebraic and rank CO classifications yield $76$ and $129$ fixed types, respectively, with $4$ and $1$ continuous families, while their combination distinguishes $287$ fixed types and $4$ continuous families. At $6$ qubits, the algebraic and rank CO classifications yield $430$ and $973$ fixed types, respectively, with $32$ and $9$ continuous families, while their combination distinguishes $4926$ fixed types and $147$ continuous families.

\subsection{Beyond one- and two-body correlators}

Throughout this work we have restricted the classification to COs that can be obtained from one- and two-body correlations, since these quantities can be measured efficiently without reconstructing the full state. For $1$ through $4$ qubits, this information is already sufficient to obtain the complete CO classifications reported above. For larger systems, however, there is no reason to stop at two-body operators. Positive COs supported on a subset $S$ probes the reduced state $\ker\rho_S$ on $|S|$ qubits, so higher-body reduced states provide a direct way of extending the rank-CO construction whenever the one- and two-body data are no longer sufficient.

Then, our framework is sufficiently general that one could also consider an  extension of COs acting on multiple coherent copies of the state \cite{leifer2004measuring}. Indeed, since by Riesz representation theorem any homogeneous polynomial $f$ of degree $k$ can be written as
\begin{equation}
    f(\psi)=\langle w_f|\psi^{\otimes k}\rangle,
\end{equation}
for some vector $\ket{w_f}$ on $k$ copies. Hence, the condition $f(\psi)=0$ is equivalent to the positive zero-eigenvalue CO
\begin{equation}
    \ket{w_f}\!\bra{w_f}
\end{equation}
annihilating $\ket{\psi}^{\otimes k}$. Polynomial-invariant approaches, and similarly constructions based on polynomial covariants, can therefore also be incorporated into the CO framework when coherent copies of the state are available~\cite{luque2003polynomial,gour2013classification,walter2013entanglement,gharahi2020fine}. Together with higher-body rank COs, this provides a systematic route for going beyond the low-body classification considered here and connecting stabilizer-, rank-, and polynomial-based descriptions within the same framework.

\section{Discussion}
In this work, we developed a framework for classifying quantum states through their conserved operators. When restricted to the algebra generating the allowed transformations, the COs coincide with the projective stabilizer of the state. Their number determines the dimension of the orbit, while the structure of the algebra contains additional information that can be used for classification. In particular, the commutation relations and eigenvalues characterizing the CO algebra can rule out transformations between states whose orbits have the same dimension. We also showed that the largest CO algebras correspond to highest-weight states, which for local transformations are product states. This classification is necessarily coarser than full SLOCC equivalence, since distinct SLOCC orbits can share the same algebraic CO data and have full-rank two-body reduced states.

In addition to the previous, we also considered positive zero-eigenvalue COs. For multipartite entanglement, the local positive COs considered here probe the kernels of reduced density matrices, so their number recovers rank-based SLOCC constraints previously studied in the literature. In this way, our work allows us to combine algebraic- and rank-based classifications within the same CO framework. This unification also leads to a substantial practical refinement, since the rank COs can distinguish states with identical algebraic COs and identical reduced state ranks. Importantly, all the information entering our classification is contained in one- and two-body correlations, allowing local conversion to be ruled out without full state reconstruction.

This low-body character is one of the main advantages of the resulting classification. In fact, recent work has shown that complete one- and two-point correlations are already sufficient to determine continuous strong symmetries and their characters~\cite{yan2026strong}. Thus, by combining this algebraic information with the kernels of two-body reduced states, we obtain a substantially finer classification of state convertibility while using the same level of correlation data. To our knowledge, this gives the most refined exhaustive classification of SLOCC-convertibility constraints through six qubits obtained entirely from one- and two-body correlations.  Higher-body and multi-copy COs, discussed above, provide a natural route for refining this description further.

\section{Methods}
\begingroup
\setlength{\parskip}{1.5pt}
\linespread{0.98}\selectfont

\subsection{Fermionic representation of rank COs}

Here we show the notable fact that the algebraic and rank COs can be combined into a single one-body fermionic description. More than a curiosity, this connection sheds light on the role of the rank COs. To begin, let us associate to each qudit $p$, $d_p$ fermionic modes, and let
\begin{equation}
 \iota:\bigotimes_{p=1}^n \C^{d_p}\longrightarrow\mathcal F_{d}
\end{equation}
denote the embedding into the single-occupancy sector of the $d=\sum_p d_p$-mode fermionic Fock space such that~\cite{levay2015embedding}
\begin{equation}
\iota\ket{i_1,\ldots,i_n} = c_{1i_1}^\dagger\cdots c_{ni_n}^\dagger
 \ket{\mathrm{vac}}.
\label{eq:fermionic-embedding-methods}
\end{equation}
The embedded states contain exactly one fermion in each $d_p$ modes associated with a qudit, such that $\hat{N}_p \iota| \psi\rangle = \iota| \psi\rangle$, with
\begin{subequations}
\begin{equation}
\hat{N}_p = \sum_i G_{pp}^{ii}, \quad p=1,\ldots,n,
\label{numberop}
\end{equation}
and $G_{pq}^{ij}=c_{pi}^\dagger c_{qj}$. Then, one can see that a local matrix ${\bf u}$ acting on $\mathcal{H}_p$ is represented by the number-preserving one-body operator $\sum_{ij} u_{ij} G_{pp}^{ij}$, and then the algebraic COs in the fermionic framework have the form
\begin{equation}
Q^{\rm dis} = \sum_p \sum_{ij} u^p_{ij} G_{pp}^{ij}, \quad \sum_i u^p_{ii} = 0.
\label{Qdis}
\end{equation}

In general, we can consider the covariance matrix associated with the full set of one-body operators $G_{pq}^{ij}$. Since $\iota|\psi\rangle$ contains exactly one fermion in each qubit sector, this covariance is block diagonal with respect to the sector-transfer
pattern,
\begin{equation}
{\bf C}(\iota|\psi\rangle) = {\bf C}^{\rm dis}(|\psi\rangle) \oplus \bigoplus_{p\neq q} {\bf C}^{pq}(\iota|\psi\rangle). \nonumber
\label{eq:fermionic-covariance-blocks}
\end{equation}
Here, ${\bf C}^{\rm dis}$ is the covariance of the intrasector operators
$G_{pp}^{ij}$, whose null space leads to the trivial COs \eqref{numberop} and to the algebraic COs \eqref{Qdis}, while ${\bf C}^{pq}$ is the $d_p d_q\times d_p d_q$ covariance block of the operators $G_{pq}^{ij}$ for fixed $p\neq q$, whose null space leads to operators of the form
\begin{equation}
Q^{pq} =\sum_{i,j} u^{pq}_{ij}G_{pq}^{ij}.
\label{eq:transfer-methods}
\end{equation}
\end{subequations}

Finally, the complete one-body fermionic CO space for the states $\iota |\psi\rangle$ decomposes as
\begin{equation}
\mathfrak p[\iota \psi] = \operatorname{span}\{\hat N_p\}_{p=1}^n \oplus \mathfrak p_{\rm dis}[\iota \psi]
\oplus \bigoplus_{p\neq q} \mathfrak p_{pq}[\iota \psi],
\label{eq:fermionic-co-decomposition}
\end{equation}
where $\mathfrak p_{\rm dis}[\psi]$ is the subspace spanned by the intrasectors COs of the form \eqref{Qdis}, while $\pstab_{pq}[\psi]$ is the subspace spanned by the intersector COs of the form \eqref{eq:transfer-methods}. Each of the subspaces appearing in Eq. \eqref{eq:fermionic-co-decomposition} is itself a Lie algebra: $\mathfrak p_{\rm dis}[\psi]$ is the CO algebra in the distinguishable representation, while $\operatorname{span}\{\hat N_p\}_{p=1}^n$ and $\mathfrak p_{pq}[\psi]$ are Abelian algebras (i.e., all of its elements mutually commute). However, elements belonging to the different subalgebras need not commute with one another. The three relevant intersector commutation relations are
\begin{subequations}
\begin{eqnarray}
\left[Q^{\rm dis}({\bf u}^p), Q^{pq}(\bf u) \right]
&=& Q^{pq}({\bf u}^p{\bf u}-{\bf u}{\bf u}^q), \\
\left[Q^{pq}({\bf u}),Q^{qp}({\bf v})\right] &=&
Q^{pp}({\bf u}{\bf v})-Q^{qq}({\bf v}{\bf u}), \\
\left[ Q^{pq}({\bf u}), Q^{qr}({\bf v})\right]
&=& Q^{pr}({\bf u}{\bf v}),
\end{eqnarray}
\label{eq:fermionic-sector-commutators}
\end{subequations}
\!\!for $p \neq q \neq r$. We also note that the intersector CO spaces transform covariantly under
local transformations for arbitrary local dimensions. If $|\phi\rangle = (g_1\otimes\cdots\otimes g_n)|\psi\rangle$ then
\begin{equation}
\mathfrak p_{pq}[\iota \phi] =
g_p\,\mathfrak p_{pq}[\iota \psi]\,g_q^{-1}.
\label{eq:intersector-transformation}
\end{equation}
Thus, the relative geometry of the intersector CO spaces is preserved along the SLOCC orbit, and together with Proposition 2, Eq. \eqref{eq:intersector-transformation} provides additional orbit constraints that can further refine the algebraic classification. 

However, for qubits, we now show that each intersector space
$\mathfrak p_{pq}[\psi]$ is linearly isomorphic to
$\ker\rho_{pq}[\psi]$, and therefore contains exactly the same information as the corresponding pair-kernel data.

Using the bipartite decomposition $pq$-rest, 
\begin{equation}
 \ket\psi
 =\sum_{i,j=0}^1\ket{i_pj_q} \ket{\phi_{ij}},
\end{equation}
a direct calculation shows that $Q^{pq} \iota\ket\psi=0$ is equivalent to
\begin{equation}
\sum_{i,j=0}^1\beta_{ij}^{pq}\ket{\phi_{ij}}=0,
 \label{eq:pair-contraction-methods}
\end{equation}
for ${\bm \beta}^{pq} = -i {\bm \sigma}_y {\bf u}^{pq}$, which is precisely the condition for the pair ket
$\ket{\beta_{pq}} = \sum_{ij} \beta^{pq*}_{ij}\ket{ij}$ to belong to $\ker \rho_{pq}[\psi]$. Thus, it follows that $\pstab_{pq} [\psi] \sim \ker \rho_{pq}[\psi]$ and in particular,
\begin{equation}
\dim \pstab_{pq}[\psi] = \dim \ker  \rho_{pq}[\psi] = 4- {\rm rank} \rho_{pq}[\psi].
\end{equation}
The algebraic and rank COs can therefore be viewed as complementary parts of the same one-body fermionic CO space, with the former acting within the modes associated with each qubit and the latter transferring a fermion between different qubits.

\subsection{Classification of the algebraic and rank CO data in qubit systems}

We now describe how the exhaustive classifications reported in Theorems \ref{thm:stabilizer-classification}, \ref{thm:rank-classification} and \ref{thm:classification} were obtained. 

To begin, the algebraic classification
retains the pair
\begin{equation}
\mathcal D_{\mathrm{alg}}(\psi) =
\bigl(\pstab[\psi],\chi_\psi\bigr),
\end{equation}
with two such data sets identified when they are related by a single
invertible local transformation and a permutation of the qubits, as
in Proposition~\ref{prop:co-equiv}.

The first step is to put the local CO algebra in a basis adapted to
its algebraic structure. Then, we partition the set of $n$ sites into $K+1$ disjoint blocks $\mathcal P_1,\ldots,\mathcal P_K$ and $\mathcal{\bar P}$, and we write 
\begin{equation}
\pstab[\psi] = \bigoplus_{k=1}^K \pstab_{k}[\psi] \oplus \bar {\pstab}[\psi], 
\label{eq:adapted-methods}
\end{equation}
where each $\mathfrak p_k[\psi]$ is a collective $\mathfrak{sl}_2$ algebra
supported on the block $\mathcal P_k$,
\begin{equation}
\mathfrak p_k[\psi] =
{\rm span} \left\{\sigma_{\mathcal P_k}^z,\, \sigma_{\mathcal P_k}^+,\, \sigma_{\mathcal P_k}^- \right\}, \quad
\sigma_{\mathcal P_k}^{\mu} = \sum_{p\in\mathcal P_k} \sigma_p^\mu, 
\end{equation}
while $\bar{\mathfrak p}[\psi]$ contains the remaining COs, for which the nilpotent part can be chosen entirely in the positive-root sector by local Weyl transformations, such that
\begin{equation}
\bar{\mathfrak p}[\psi] = \bar{\mathfrak p}^z[\psi] \oplus \bar{\mathfrak p}^+[\psi], \quad \bar{\mathfrak p}^\mu[\psi] \subseteq \bar{\mathfrak g}^\mu,
\end{equation}
with $\bar{\mathfrak g}^\mu = {\rm span}\{ \sigma_{\bar p}^\mu \ \forall \  {\bar p} \in \mathcal{\bar P}\}$. The complete proof of Eq. \eqref{eq:adapted-methods} and the
construction of the corresponding local basis are given in the SI.

In this adapted basis, the diagonal sector can be enumerated
directly from the computational support of the state. We denote ${\bm x}=(x_1,\ldots,x_n)\in\{0,1\}^n$ and
$|{\bm x}\rangle\equiv|x_1\cdots x_n\rangle$. Thus, if
\begin{equation}
\ket\psi=\sum_{\bm x} \psi_{\bm x} \ket {\bm x}, \quad {\bm s}({\bm x})=((-1)^{x_1},\ldots,(-1)^{x_n}),
\end{equation}
and $Q^z=\sum_p \epsilon_p\sigma_p^z$, then $Q^z|\psi\rangle = \lambda |\psi\rangle$ if and only if
\begin{equation}
{\bm s}({\bm x})\cdot {\bm \epsilon}= \lambda, \quad \forall \ {\bm x} : \psi_{\bm x} \neq 0.
 \label{eq:diagonal-support-methods}
\end{equation}
Each computational basis state $|{\bm x}\rangle$ defines a sign vector
${\bm s}({\bm x})\in\{\pm1\}^n$, i.e., a vertex of the $n$-dimensional hypercube $\{\pm1\}^n$. Hence, the possible Cartan sectors $\mathfrak p^z{[\psi]}$ are therefore determined by the affine spans of subsets of these vertices. When applying this construction to the decomposition in Eq. \eqref{eq:adapted-methods}, the
Cartan generators associated with the collective blocks $\mathcal P_k$
are already fixed. Hence, only the residual Cartan sector supported on
$\bar P$ needs to be enumerated.

Once $\bar{\mathfrak p}^{\,z}[\psi]$ is fixed, the sites in $\bar P$ are grouped according to their weight under its adjoint action. The raising sector $\bar{\mathfrak p}^{\,+}[\psi]$ must decompose within these equal-weight blocks. We enumerate the possible subspaces in reduced row echelon form, use the residual diagonal local transformations to remove coefficients that can be fixed without loss of generality, and retain all remaining parameters and zero patterns. Finally, Lie closure is imposed, and descriptions related by the remaining local transformations and qubit permutations are identified. 

When moving to the refined classification, each algebraic CO type is augmented with all the rank data. That is, we study the data in the spaces
\begin{equation}
\mathcal D(\psi) = \left(\pstab[\psi],\chi_\psi, \{\ker \rho_{pq}[\psi]\}_{p<q} \right).
\label{eq:full-data-methods}
\end{equation}
Equivalently, we may work with the fermionic intersector spaces $\pstab_{pq}[\psi]$. The crucial point is that the pair spaces are not classified independently. The same local matrix $g_p$ must act on every pair containing qubit $p$. Hence, relative directions shared by different pairs must be retained. In the fermionic description, their compatibility with the local algebra and with one another follows
from relations \eqref{eq:fermionic-sector-commutators}, which strongly restrict which collections of pair spaces can coexist.

For each algebraic CO type we enumerate all compatible pair spaces, retain every parameter locus on which their dimensions or relative
geometry changes, and then quotient the complete data in Eq.~\eqref{eq:full-data-methods} by one simultaneous set of local transformations and qubit permutations. Performing separate basis
changes on different pairs would remove precisely the relative information that allows the rank COs to refine the algebraic classification.

At this point, enumerating algebraically compatible CO data is not enough, since a proposed set of constraints may force additional COs for every state satisfying them. We therefore test whether every candidate can be realized exactly.

For a proposed algebraic CO algebra $\pstab$ with character $\chi$ and pair spaces $\mathcal{K}_{pq}$, define the common state space
\footnotesize
\begin{equation}
 W= \left\{\ket w\in \mathcal{H}\!:\!\! \begin{array}{l}
 (Q-\chi(Q)\id)\ket w=0
 \quad\forall\,Q\in\pstab,\\[2pt]
 (P_{pq}^k\otimes\id)\ket w=0
 \quad\forall\,\ket{k_{pq}}\in \mathcal{K}_{pq},\ p<q
 \end{array}
 \right\}.
 \nonumber
\end{equation}
\normalsize
If $W=\{0\}$, the candidate cannot occur. Otherwise, choose a basis
$\{\ket{w_i}\}_{i=1}^{m}$ of $W$, where $m=\dim W$, and write a general state in this space as $\ket{w}=\sum_{i=1}^{m}c_i\ket{w_i}$.

For this general state, we form the augmented local action matrix
\begin{equation}
 A(w)
 =
 \bigl[
 \sigma_1^zw,\sigma_1^+w,\sigma_1^-w,
 \ldots,
 \sigma_n^zw,\sigma_n^+w,\sigma_n^-w,
 -w
 \bigr],
 \nonumber
\end{equation}
whose kernel gives all algebraic COs of $\ket w$ together with their eigenvalues. For every pair $p,q$, we also form the pair contraction matrix
\begin{equation}
B_{pq}(w) = \bigl[\ket{v_{00}},\ket{v_{01}},
\ket{v_{10}},\ket{v_{11}}\bigr],
\nonumber
\end{equation}
where $|v_{ab}\rangle = (\langle ab|_{pq}\otimes \id_{\overline{pq}})|w\rangle$ and $a,b\in\{0,1\}$. Its kernel is the pair dual space $\mathcal K_{pq}(w)$. If the prescribed dimensions are $d=\dim \pstab$, $\Delta d_{pq} = \dim {\mathcal K}_{pq}$, then the candidate is realized with no additional algebraic or rank COs precisely when
\begin{align}
 {\rm rank}\,A(w)
 &=
 3n+1-d,
 \label{eq:occurrence-ranks-methods-1}\\
 {\rm rank}\,B_{pq}(w)
 &=
 4-\Delta d_{pq}
 \qquad \forall\,p<q.
 \label{eq:occurrence-ranks-methods-2}
\end{align}
These are the largest ranks compatible with the prescribed constraints. Any smaller rank would imply the presence of additional COs not included in the candidate.

Since $\ket{w}$ depends linearly on the $c_i$, the entries of $A(w)$ and $B_{pq}(w)$ are also linear in these coordinates. Their minors are therefore homogeneous polynomials in $\bm c$. Attaining a required rank is equivalent to at least one minor of the corresponding size being nonzero, and hence defines a Zariski-open subset of $W$. If the target rank is attainable for every required matrix, each of these open sets is nonempty. Since $W$ is a complex vector space, their finite intersection is also nonempty, and there exists a single state $\ket{w}$ that realizes all prescribed ranks simultaneously. Conversely, if every minor of the required size vanishes identically on $W$, the corresponding additional CO is unavoidable for every state satisfying the candidate constraints.

This also gives a finite exact test. Suppose the product of one nonvanishing minor for each required rank has total degree at most $D$. A nonzero polynomial of degree at most $D$ cannot vanish on every point of
\begin{equation}
 \{0,1,\ldots,D\}^{m},
\end{equation}
with $m=\dim W$. Thus, whenever the candidate is realizable, an exact witness can be found by a finite search over these coordinates. This gives a computable test without restricting the analysis to a particular trial state.

Continuous families are treated in the same way, but with the family parameters kept symbolic.  For each fixed parameter value, the realization problem is again determined by polynomial rank conditions in $\bm c$. The parameter domain must then be divided whenever the rank of a constraint, a pivot, a denominator, or an equivalence relation changes. In this way, the required ranks are verified over the complete parameter domain rather than only at sampled values. 

Finally, realized candidates are quotiented only after checking a single physical equivalence acting simultaneously on the algebraic CO algebra, its character, and all pair spaces. Two candidates are therefore identified only when the same local transformation and qubit permutation maps all of these data at once. The complete classification then requires three steps. We first construct a cover of all admissible candidate data, then test the exact realization of every candidate, and finally quotient the realized data by the allowed local transformations and qubit permutations.

Importantly, the algorithm also provides a direct procedure for classifying a given quantum state. Once the exhaustive catalog has been constructed, the algebraic COs and their characters can be extracted from the covariance matrix, while the rank COs follow from the kernels of the two-body reduced density matrices. The resulting data can then be matched to the corresponding family in the catalog, up to simultaneous local transformations and qubit permutations. Thus, identifying the CO family of an arbitrary state requires only one- and two-body correlations, without repeating the exhaustive enumeration.

\subsection{Complete classification results}

The complete algebraic-plus-rank CO classification through $4$ qubits, together with representative states for each class, is presented in Table \ref{tabgen}. 

\newsavebox{\MainCellBox}
\newlength{\MainRowHeight}

\newcommand{\MainPad}[1]{%
  \sbox{\MainCellBox}{\strut #1}%
  \raisebox{0pt}%
    [\dimexpr\ht\MainCellBox+1.8pt\relax]%
    [\dimexpr\dp\MainCellBox+1.8pt\relax]%
    {\usebox{\MainCellBox}}%
}
\newcommand{\MainGroupHeight}[3]{%
  \global\MainRowHeight=0pt
  \MainMeasureRow{#1}\MainMeasureRow{#2}\MainMeasureRow{#3}%
}
\newcommand{\MainMeasureRow}[1]{%
  \sbox{\MainCellBox}{\strut #1}%
  \ifdim\dimexpr\ht\MainCellBox+\dp\MainCellBox+3.6pt\relax>\MainRowHeight
    \global\MainRowHeight=\dimexpr\ht\MainCellBox+\dp\MainCellBox+3.6pt\relax
  \fi
}
\newcommand{\MainSpan}[2]{%
  \raisebox{\dimexpr
    (#1\MainRowHeight-\MainRowHeight)/2+
    (#1\arrayrulewidth-\arrayrulewidth)/2\relax
  }[0pt][0pt]{#2}%
}

\begin{table*}[p]
\centering
\scriptsize
\setlength{\tabcolsep}{1.3pt}
\renewcommand{\arraystretch}{0}

\ifdefined\MainCellBox\else
  \newsavebox{\MainCellBox}
\fi
\ifdefined\MainRowHeight\else
  \newlength{\MainRowHeight}
\fi

\providecommand{\MainPad}[1]{%
  \sbox{\MainCellBox}{\strut #1}%
  \raisebox{0pt}%
    [\dimexpr\ht\MainCellBox+1.8pt\relax]%
    [\dimexpr\dp\MainCellBox+1.8pt\relax]%
    {\usebox{\MainCellBox}}%
}
\providecommand{\MainGroupHeight}[3]{%
  \global\MainRowHeight=0pt
  \MainMeasureRow{#1}\MainMeasureRow{#2}\MainMeasureRow{#3}%
}
\providecommand{\MainMeasureRow}[1]{%
  \sbox{\MainCellBox}{\strut #1}%
  \ifdim\dimexpr\ht\MainCellBox+\dp\MainCellBox+3.6pt\relax>\MainRowHeight
    \global\MainRowHeight=\dimexpr\ht\MainCellBox+\dp\MainCellBox+3.6pt\relax
  \fi
}
\providecommand{\MainSpan}[2]{%
  \raisebox{\dimexpr
    (#1\MainRowHeight-\MainRowHeight)/2+
    (#1\arrayrulewidth-\arrayrulewidth)/2\relax
  }[0pt][0pt]{#2}%
}

\providecommand{\MainPartialRule}{%
  \hhline{%
    >{\arrayrulecolor{black}}|>{\arrayrulecolor{gray!15}}-%
    >{\arrayrulecolor{black}}|>{\arrayrulecolor{black}}-%
    >{\arrayrulecolor{black}}|>{\arrayrulecolor{red!10}}-%
    >{\arrayrulecolor{black}}|>{\arrayrulecolor{red!10}}-%
    >{\arrayrulecolor{black}}|>{\arrayrulecolor{red!10}}-%
    >{\arrayrulecolor{black}}|>{\arrayrulecolor{red!10}}-%
    >{\arrayrulecolor{black}}|>{\arrayrulecolor{yellow!15}}-%
    >{\arrayrulecolor{black}}|>{\arrayrulecolor{yellow!15}}-%
    >{\arrayrulecolor{black}}|>{\arrayrulecolor{green!10}}-%
    >{\arrayrulecolor{black}}|>{\arrayrulecolor{green!10}}-%
    >{\arrayrulecolor{black}}|%
  }%
}

\resizebox{0.99\linewidth}{!}{%
\begin{tabular}{
|>{\columncolor{gray!15}[\tabcolsep][\tabcolsep]}w{c}{9pt}
|>{\columncolor{white!10}[\tabcolsep][\tabcolsep]}w{c}{187.70851pt}
|>{\columncolor{red!10}[\tabcolsep][\tabcolsep]}w{c}{74.78426pt}
|>{\columncolor{red!10}[\tabcolsep][\tabcolsep]}w{c}{67.87558pt}
|>{\columncolor{red!10}[\tabcolsep][\tabcolsep]}w{c}{52.7562pt}
|>{\columncolor{red!10}[\tabcolsep][\tabcolsep]}w{c}{11.07146pt}
|>{\columncolor{yellow!15}[\tabcolsep][\tabcolsep]}w{c}{134.43834pt}
|>{\columncolor{yellow!15}[\tabcolsep][\tabcolsep]}w{c}{14.58401pt}
|>{\columncolor{green!10}[\tabcolsep][\tabcolsep]}w{c}{7.97226pt}
|>{\columncolor{green!10}[\tabcolsep][\tabcolsep]}w{c}{14.83345pt}|}

\noalign{\global\arrayrulewidth=.8pt}
\hhline{~|---------|}
\noalign{\global\arrayrulewidth=.4pt}

\multicolumn{1}{c|}{} &
& \multicolumn{4}{>{\columncolor{red!15}[\tabcolsep][\tabcolsep]}c|}{%
    \MainPad{\bfseries\footnotesize Algebraic/Intrasector COs}}
& \multicolumn{2}{>{\columncolor{yellow!20}[\tabcolsep][\tabcolsep]}c|}{%
    \MainPad{\bfseries\footnotesize Rank/Intersector COs}}
& \cellcolor{green!15}
& \cellcolor{green!15} \\
\hhline{~|~|----|--|~|~|}

\multicolumn{1}{c|}{}
& \raisebox{6.6pt}[0pt][0pt]{\bfseries\footnotesize Representative state}
& \cellcolor{red!15}\MainPad{\bfseries\footnotesize Cartan}
& \cellcolor{red!15}\MainPad{\bfseries\footnotesize Non-Cartan}
& \cellcolor{red!15}\MainPad{\bfseries\footnotesize Algebra}
& \cellcolor{red!15}\MainPad{\bfseries\footnotesize $D_{\pstab}$}
& \cellcolor{yellow!20}\MainPad{\bfseries\footnotesize COs ($p<q$)}
& \cellcolor{yellow!20}\MainPad{\bfseries\footnotesize $\tfrac{\Delta D}{2}$}
& \cellcolor{green!15}\raisebox{6.6pt}[0pt][0pt]{\bfseries\footnotesize $D$}
& \cellcolor{green!15}\raisebox{6.6pt}[0pt][0pt]{\bfseries\footnotesize deg} \\
\noalign{\global\arrayrulewidth=.8pt}
\hhline{~|---------|}
\noalign{\global\arrayrulewidth=.4pt}

\multicolumn{1}{c|}{} &
\multicolumn{9}{>{\columncolor{cyan!15}[\tabcolsep][\tabcolsep]}c|}{%
  \MainPad{\bfseries $n=1$ qubit}} \\ \hline

\MainPad{1}
& \MainPad{$|0\rangle$}
& \MainPad{$\sigma_1^z(1)$}
& \MainPad{$\sigma_1^+$}
& \MainPad{$\mathfrak b$}
& \MainPad{2}
& \MainPad{---}
& \MainPad{0}
& \MainPad{2}
& \MainPad{$\times1$} \\ \hline

\multicolumn{1}{c|}{} &
\multicolumn{9}{>{\columncolor{cyan!15}[\tabcolsep][\tabcolsep]}c|}{%
  \MainPad{\bfseries $n=2$ qubits}} \\ \hline

\MainPad{1}
& \MainPad{$|00\rangle$}
& \MainPad{$\sigma_p^z(1)$}
& \MainPad{$\sigma_p^+$}
& \MainPad{$\mathfrak b^{\oplus2}$}
& \MainPad{4}
& \MainPad{$G_{12}^{0i},\,G_{12}^{11}$}
& \MainPad{3}
& \MainPad{10}
& \MainPad{$\times1$} \\ \hline

\MainPad{2}
& \MainPad{$|\Psi^-_{12}\rangle$}
& \MainPad{$\sigma_1^z+\sigma_2^z$}
& \MainPad{$\sigma_1^\pm+\sigma_2^\pm$}
& \MainPad{$\mathfrak{sl}_2$}
& \MainPad{3}
& \MainPad{$G_{12}^{01},G_{12}^{10},G_{12}^{00}\!-\!G_{12}^{11}$}
& \MainPad{3}
& \MainPad{9}
& \MainPad{$\times1$} \\ \hline

\multicolumn{1}{c|}{} &
\multicolumn{9}{>{\columncolor{cyan!15}[\tabcolsep][\tabcolsep]}c|}{%
  \MainPad{\bfseries $n=3$ qubits}} \\ \hline

\MainPad{1}
& \MainPad{$|000\rangle$}
& \MainPad{$\sigma_p^z(1)$}
& \MainPad{$\sigma_p^+$}
& \MainPad{$\mathfrak b^{\oplus3}$}
& \MainPad{6}
& \MainPad{$G_{pq}^{0i},\,G_{pq}^{11}$}
& \MainPad{9}
& \MainPad{24}
& \MainPad{$\times1$} \\ \hline

\MainPad{2}
& \MainPad{$|0\rangle|\Psi^-_{23}\rangle$}
& \MainPad{$\sigma_1^z(1),\,\sigma_2^z+\sigma_3^z$}
& \MainPad{$\sigma_1^+,\,\sigma_2^\pm+\sigma_3^\pm$}
& \MainPad{$\mathfrak b\oplus\mathfrak{sl}_2$}
& \MainPad{5}
& \MainPad{$G_{1q}^{0i},\,G_{23}^{01},G_{23}^{10},G_{23}^{00}\!-\!G_{23}^{11}$}
& \MainPad{7}
& \MainPad{19}
& \MainPad{$\times3$} \\ \hline

\MainPad{3}
& \MainPad{$|W_3\rangle$}
& \MainPad{$\sum_p\sigma_p^z(1)$}
& \MainPad{$\sigma_{1,2}^+-\sigma_{2,3}^+$}
& \MainPad{$\mathbb C\ltimes\mathbb C^2$}
& \MainPad{3}
& \MainPad{$G_{pq}^{01},\,G_{pq}^{00}+G_{pq}^{11}$}
& \MainPad{6}
& \MainPad{15}
& \MainPad{$\times1$} \\ \hline

\MainPad{4}
& \MainPad{$|{\rm GHZ}_3\rangle$}
& \MainPad{$\sigma_{1,2}^z-\sigma_{2,3}^z$}
& \MainPad{---}
& \MainPad{$\mathbb C_{\rm tor}^{2}$}
& \MainPad{2}
& \MainPad{$G_{pq}^{ii}$}
& \MainPad{6}
& \MainPad{14}
& \MainPad{$\times1$} \\ \hline

\multicolumn{1}{c|}{} &
\multicolumn{9}{>{\columncolor{cyan!15}[\tabcolsep][\tabcolsep]}c|}{%
  \MainPad{\bfseries $n=4$ qubits}} \\ \hline

\MainPad{1}
& \MainPad{$|0000\rangle$}
& \MainPad{$\sigma_p^z(1)$}
& \MainPad{$\sigma_p^+$}
& \MainPad{$\mathfrak b^{\oplus4}$}
& \MainPad{8}
& \MainPad{$G_{pq}^{0i},\,G_{pq}^{11}$}
& \MainPad{18}
& \MainPad{44}
& \MainPad{$\times1$} \\ \hline

\MainPad{2}
& \MainPad{$|00\rangle|\Psi^-_{34}\rangle$}
& \MainPad{$\sigma_{1,2}^z(1),\,\sigma_3^z+\sigma_4^z$}
& \MainPad{$\sigma_{1,2}^+,\,\sigma_3^\pm+\sigma_4^\pm$}
& \MainPad{$\mathfrak b^{\oplus2}\oplus\mathfrak{sl}_2$}
& \MainPad{7}
& \MainPad{$G_{12}^{11},G_{1q}^{0i},G_{2q}^{0i},G_{34}^{01},G_{34}^{10},G_{34}^{00}\!-\!G_{34}^{11}$}
& \MainPad{14}
& \MainPad{35}
& \MainPad{$\times6$} \\ \hline

\MainPad{3}
& \MainPad{$|0\rangle|W_3\rangle$}
& \MainPad{$\sigma_1^z(1),\,\sum_{p=2}^4\sigma_p^z(1)$}
& \MainPad{$\sigma_1^+,\,\sigma_{2,3}^+-\sigma_{3,4}^+$}
& \MainPad{$\mathfrak b\oplus(\mathbb C\ltimes\mathbb C^2)$}
& \MainPad{5}
& \MainPad{$G_{1q}^{0i},G_{pq}^{01},\,G_{pq}^{00}+G_{pq}^{11}\ (p,q=2,3,4)$}
& \MainPad{12}
& \MainPad{29}
& \MainPad{$\times4$} \\ \hline

\MainPad{4}
& \MainPad{$|0\rangle|{\rm GHZ}_3\rangle$}
& \MainPad{$\sigma_1^z(1),\,\sigma_{2,3}^z-\sigma_{3,4}^z$}
& \MainPad{$\sigma_1^+$}
& \MainPad{$\mathfrak b\oplus\mathbb C_{\rm tor}^{2}$}
& \MainPad{4}
& \MainPad{$G_{1q}^{0i},G_{pq}^{ii}\ (p,q=2,3,4)$}
& \MainPad{12}
& \MainPad{28}
& \MainPad{$\times4$} \\ \hline

\MainPad{5}
& \MainPad{$|W_4\rangle$}
& \MainPad{$\sum_p\sigma_p^z(2)$}
& \MainPad{$\sigma_{1,2,3}^+-\sigma_{2,3,4}^+$}
& \MainPad{$\mathbb C\ltimes\mathbb C^3$}
& \MainPad{4}
& \MainPad{$G_{pq}^{01},G_{pq}^{00}+G_{pq}^{11}$}
& \MainPad{12}
& \MainPad{28}
& \MainPad{$\times1$} \\ \hline

\MainPad{6}
& \MainPad{$|{\rm GHZ}_4\rangle$}
& \MainPad{$\sigma_{1,2,3}^z-\sigma_{2,3,4}^z$}
& \MainPad{---}
& \MainPad{$\mathbb C_{\rm tor}^{3}$}
& \MainPad{3}
& \MainPad{$G_{pq}^{00},G_{pq}^{11}$}
& \MainPad{12}
& \MainPad{27}
& \MainPad{$\times1$} \\ \hline

\MainPad{7}
& \MainPad{$|1100\rangle+|00\rangle|\Psi^-_{34}\rangle$}
& \MainPad{$\sigma_1^z-\sigma_2^z,\,\sum_{p=2}^4\sigma_p^z(1)$}
& \MainPad{$\sigma_3^++\sigma_4^+$}
& \MainPad{$\mathbb C\oplus\mathfrak b$}
& \MainPad{3}
& \MainPad{$G_{12}^{ii},G_{34}^{00}-G_{34}^{11},G_{pq}^{01}\ ((p,q)\neq(1,2))$}
& \MainPad{8}
& \MainPad{19}
& \MainPad{$\times6$} \\ \hline

\MainPad{8}
& \MainPad{$|\Psi^-_{12}\rangle|\Psi^-_{34}\rangle$}
& \MainPad{$\sigma_1^z+\sigma_2^z,\,\sigma_3^z+\sigma_4^z$}
& \MainPad{$\sigma_1^\pm+\sigma_2^\pm,\,\sigma_3^\pm+\sigma_4^\pm$}
& \MainPad{$\mathfrak{sl}_2\oplus\mathfrak{sl}_2$}
& \MainPad{6}
& \MainPad{$G_{pq}^{01},G_{pq}^{10},G_{pq}^{00}\!-\!G_{pq}^{11}\;(p,q=1,2\lor3,4)$}
& \MainPad{6}
& \MainPad{18}
& \MainPad{$\times3$} \\ \hline

\MainPad{9}
& \MainPad{$2|\Psi^+_{12}\rangle|\Psi^+_{34}\rangle-|1010\rangle$}
& \MainPad{$\sigma_1^z+\sigma_2^z,\,\sigma_3^z+\sigma_4^z$}
& \MainPad{---}
& \MainPad{$\mathbb C_{\rm tor}^{2}$}
& \MainPad{2}
& \MainPad{$G_{12}^{01},G_{12}^{10},G_{34}^{01},G_{34}^{10},G_{13}^{01},G_{14}^{00},G_{23}^{11},G_{24}^{10}$}
& \MainPad{8}
& \MainPad{18}
& \MainPad{$\times3$} \\ \hline

\MainPad{10}
& \MainPad{$|\Psi^-_{12}\rangle|0_30_4\rangle-|0_1 1_4\rangle|\Psi^-_{23}\rangle$}
& \MainPad{$\sigma_1^z+\sigma_2^z+\sigma_3^z(1)$}
& \MainPad{$\sigma_1^++\sigma_2^++\sigma_3^+$}
& \MainPad{$\mathfrak b$}
& \MainPad{2}
& \MainPad{$G_{pq}^{01}(p,q\neq4),G_{14}^{01},G_{34}^{00},G_{24}^{01}-G_{24}^{00}$}
& \MainPad{6}
& \MainPad{14}
& \MainPad{$\times4$} \\ \hline

\MainPad{11}
& \MainPad{$|\Psi^+_{12,34}\rangle+|\Psi^+_{13,24}\rangle+|\Psi^+_{14,23}\rangle$}
& \MainPad{$\sum_p\sigma_p^z$}
& \MainPad{---}
& \MainPad{$\mathbb C_{\rm tor}$}
& \MainPad{1}
& \MainPad{$G_{pq}^{00}+G_{pq}^{11}$}
& \MainPad{6}
& \MainPad{13}
& \MainPad{$\times1$} \\ \hline

\MainPad{12}
& \MainPad{$|0\rangle|W_3\rangle+|1111\rangle$}
& \MainPad{$2\sigma_1^z-\sum_{p=2}^4\sigma_p^z(1)$}
& \MainPad{---}
& \MainPad{$\mathbb C_{\rm tor}$}
& \MainPad{1}
& \MainPad{$G_{1q}^{00},G_{pq}^{00}+G_{pq}^{11}\ (p,q=2,3,4)$}
& \MainPad{6}
& \MainPad{13}
& \MainPad{$\times4$} \\ \hline

\MainPad{13}
& \MainPad{$|\Psi^+_{12,34}\rangle+|\Psi^+_{12}\rangle|0_31_4\rangle$}
& \MainPad{$\sum_p\sigma_p^z$}
& \MainPad{---}
& \MainPad{$\mathbb C_{\rm tor}$}
& \MainPad{1}
& \MainPad{$G_{12}^{00}+G_{12}^{11},G_{13}^{01},G_{23}^{01},G_{14}^{10},G_{24}^{10},G_{34}^{00}$}
& \MainPad{6}
& \MainPad{13}
& \MainPad{$\times6$} \\ \hline

\MainPad{14}
& \MainPad{$|{\rm GHZ}_4\rangle+\tfrac{\sqrt3}{\sqrt2}|0\rangle|W_3\rangle$}
& \MainPad{---}
& \MainPad{---}
& \MainPad{$0$}
& \MainPad{0}
& \MainPad{$G_{1q}^{00},G_{pq}^{00}+G_{pq}^{11}\;(p,q=2,3,4)$}
& \MainPad{6}
& \MainPad{12}
& \MainPad{$\times4$} \\ \hline

\noalign{\MainGroupHeight
  {$\alpha|{\rm GHZ}_4\rangle+\left(|\Psi^+_{12,34}\rangle+|\Psi^+_{13,24}\rangle+|\Psi^+_{14,23}\rangle\right)^{\rm a}$}
  {$|{\rm GHZ}_4\rangle+|\Psi^+_{12,34}\rangle+|\Psi^+_{13,24}\rangle+\beta|0110\rangle$}
  {}}
&
\rule[-\dimexpr\MainRowHeight/2-2.5pt\relax]{0pt}{\MainRowHeight}
\MainPad{$\alpha|{\rm GHZ}_4\rangle+\left(|\Psi^+_{12,34}\rangle+|\Psi^+_{13,24}\rangle+|\Psi^+_{14,23}\rangle\right)^{\rm a}$}
& & & & & & & & \\
\MainPartialRule

\MainSpan{2}{15}
& \rule[-\dimexpr\MainRowHeight/2-2.5pt\relax]{0pt}{\MainRowHeight}
  \MainPad{$|{\rm GHZ}_4\rangle+|\Psi^+_{12,34}\rangle+|\Psi^+_{13,24}\rangle+\beta|0110\rangle$}
& \MainSpan{2}{---}
& \MainSpan{2}{---}
& \MainSpan{2}{$0$}
& \MainSpan{2}{0}
& \MainSpan{2}{$G_{pq}^{00}+G_{pq}^{11}$}
& \MainSpan{2}{6}
& \MainSpan{2}{12}
& \MainSpan{2}{$\times1$} \\ \hline

\MainPad{16}
& \MainPad{$|\Psi^-_{12}\rangle|\Psi^-_{34}\rangle+\gamma|0011\rangle$}
& \MainPad{$\sum_p\sigma_p^z$}
& \MainPad{$\sigma_1^++\sigma_2^+,\,\sigma_3^-+\sigma_4^-$}
& \MainPad{$\mathbb C\ltimes(\mathbb C_+\oplus\mathbb C_-)$}
& \MainPad{3}
& \MainPad{$G_{12}^{01},G_{12}^{11}-G_{12}^{00},G_{34}^{10},G_{34}^{11}-G_{34}^{00}$}
& \MainPad{4}
& \MainPad{11}
& \MainPad{$\times3$} \\ \hline

\MainPad{17}
& \MainPad{$2|\Psi^+_{12}\rangle|\Psi^+_{34}\rangle-\delta^2|1010\rangle^{\rm b}$}
& \MainPad{$\sigma_1^z+\sigma_2^z,\,\sigma_3^z+\sigma_4^z$}
& \MainPad{---}
& \MainPad{$\mathbb C_{\rm tor}^{2}$}
& \MainPad{2}
& \MainPad{$G_{12}^{01},G_{12}^{10},G_{34}^{01},G_{34}^{10}$}
& \MainPad{4}
& \MainPad{10}
& \MainPad{$\times3$} \\ \hline

\MainPad{18}
& \MainPad{$|\Psi^+_{12}\rangle|\Psi^-_{34}\rangle+|1000\rangle$}
& \MainPad{$\sigma_1^z+\sigma_2^z$}
& \MainPad{$\sigma_3^++\sigma_4^+$}
& \MainPad{$\mathbb C_{\rm tor}\oplus\mathbb C_{\rm nil}$}
& \MainPad{2}
& \MainPad{$G_{12}^{01},G_{12}^{10},G_{34}^{01},G_{34}^{11}-G_{34}^{00}$}
& \MainPad{4}
& \MainPad{10}
& \MainPad{$\times6$} \\ \hline

\MainPad{19}
& \MainPad{$|0011\rangle+|1100\rangle-|\Psi^+_{12}\rangle|\Psi^+_{34}\rangle$}
& \MainPad{$\sum_p\sigma_p^z$}
& \MainPad{$\sum_p\sigma_p^\pm$}
& \MainPad{$\mathfrak{sl}_2$}
& \MainPad{3}
& \MainPad{$G_{12}^{00}+G_{12}^{11},G_{34}^{00}+G_{34}^{11}$}
& \MainPad{2}
& \MainPad{7}
& \MainPad{$\times3$} \\ \hline

\MainPad{20}
& \MainPad{$|\Psi^+_{12,34}\rangle+|\Psi^+_{13,24}\rangle+\epsilon|\Psi^+_{14,23}\rangle^{\rm c}$}
& \MainPad{$\sum_p\sigma_p^z$}
& \MainPad{---}
& \MainPad{$\mathbb C_{\rm tor}$}
& \MainPad{1}
& \MainPad{$G_{14}^{00}+G_{14}^{11},G_{23}^{00}+G_{23}^{11}$}
& \MainPad{2}
& \MainPad{5}
& \MainPad{$\times3$} \\ \hline

\MainPad{21}
& \MainPad{$\sqrt2\left(|\Psi^+_{13,24}\rangle+\delta|\Psi^+_{14,23}\rangle\right) +|0011\rangle^{\rm b}$}
& \MainPad{$\sum_p\sigma_p^z$}
& \MainPad{---}
& \MainPad{$\mathbb C_{\rm tor}$}
& \MainPad{1}
& \MainPad{$G_{12}^{01},G_{34}^{10}$}
& \MainPad{2}
& \MainPad{5}
& \MainPad{$\times3$} \\ \hline

\MainPad{22}
& \MainPad{$|0000\rangle+|\Psi^-_{12}\rangle|\Psi^-_{34}\rangle+|\Psi^-_{13}\rangle|\Psi^-_{24}\rangle$}
& \MainPad{---}
& \MainPad{$\sum_p\sigma_p^+$}
& \MainPad{$\mathbb C_{\rm nil}$}
& \MainPad{1}
& \MainPad{$G_{14}^{00}+G_{14}^{11},G_{23}^{00}+G_{23}^{11}$}
& \MainPad{2}
& \MainPad{5}
& \MainPad{$\times3$} \\ \hline

\MainPad{23}
& \MainPad{$|\Psi^-_{12}\rangle|\Psi^-_{34}\rangle+|\Psi^-_{13}\rangle|0_2 0_4\rangle$}
& \MainPad{---}
& \MainPad{$\sum_p\sigma_p^+$}
& \MainPad{$\mathbb C_{\rm nil}$}
& \MainPad{1}
& \MainPad{$G_{12}^{01},G_{34}^{01}$}
& \MainPad{2}
& \MainPad{5}
& \MainPad{$\times3$} \\ \hline

\noalign{\MainGroupHeight
  {$|00\rangle|\Phi^+_{34}\rangle+|\Psi^+_{13,24}\rangle+\delta|\Psi^+_{14,23}\rangle^{\rm b}$}
  {$\eta|\Psi_{12}^{+}\rangle|\Phi^+_{34}\rangle+\sqrt2|\Psi_{12}^{-}\rangle|01\rangle_{34}+|0010\rangle$}
  {}}
&
\rule[-\dimexpr\MainRowHeight/2-2.5pt\relax]{0pt}{\MainRowHeight}
\MainPad{$|00\rangle|\Phi^+_{34}\rangle+|\Psi^+_{13,24}\rangle+\delta|\Psi^+_{14,23}\rangle^{\rm b}$}
& & & & & & & & \\
\MainPartialRule

\MainSpan{2}{24}
& \rule[-\dimexpr\MainRowHeight/2-2.5pt\relax]{0pt}{\MainRowHeight}
  \MainPad{$\eta|\Psi_{12}^{+}\rangle|\Phi^+_{34}\rangle+\sqrt2|\Psi_{12}^{-}\rangle|01\rangle_{34}+|0010\rangle$}
& \MainSpan{2}{---}
& \MainSpan{2}{---}
& \MainSpan{2}{$0$}
& \MainSpan{2}{0}
& \MainSpan{2}{$G_{12}^{01},\,G_{34}^{01}+G_{34}^{10}$}
& \MainSpan{2}{2}
& \MainSpan{2}{4}
& \MainSpan{2}{$\times6$} \\ \hline

\noalign{\MainGroupHeight
  {$\theta|{\rm GHZ}_4\rangle+|\Psi^+_{12,34}\rangle+\delta\left(|\Psi^+_{13,24}\rangle+|\Psi^+_{14,23}\rangle\right)^{\rm b,d}$}
  {$|{\rm GHZ}_4\rangle+|\Psi^+_{12,34}\rangle+\epsilon|\Psi^+_{13,24}\rangle+\iota|0110\rangle^{\rm c}$}
  {$\kappa(\sqrt2|{\rm GHZ}_4\rangle-|\Psi^+_{12}\rangle|\Psi^+_{34}\rangle)+|00\rangle|\Psi^+_{34}\rangle+|\Psi^+_{12}\rangle|11\rangle$}}
&
\rule[-\dimexpr\MainRowHeight/2-2.5pt\relax]{0pt}{\MainRowHeight}
\MainPad{$\theta|{\rm GHZ}_4\rangle+|\Psi^+_{12,34}\rangle+\delta\left(|\Psi^+_{13,24}\rangle+|\Psi^+_{14,23}\rangle\right)^{\rm b,d}$}
& & & & & & & & \\
\MainPartialRule

&
\rule[-\dimexpr\MainRowHeight/2-2.5pt\relax]{0pt}{\MainRowHeight}
\MainPad{$|{\rm GHZ}_4\rangle+|\Psi^+_{12,34}\rangle+\epsilon|\Psi^+_{13,24}\rangle+\iota|0110\rangle^{\rm c}$}
& & & & & & & & \\
\MainPartialRule

\MainSpan{3}{25}
& \rule[-\dimexpr\MainRowHeight/2-2.5pt\relax]{0pt}{\MainRowHeight}
  \MainPad{$\kappa(\sqrt2|{\rm GHZ}_4\rangle-|\Psi^+_{12}\rangle|\Psi^+_{34}\rangle)+|00\rangle|\Psi^+_{34}\rangle+|\Psi^+_{12}\rangle|11\rangle$}
& \MainSpan{3}{---}
& \MainSpan{3}{---}
& \MainSpan{3}{$0$}
& \MainSpan{3}{0}
& \MainSpan{3}{$G_{12}^{00}+G_{12}^{11},G_{34}^{00}+G_{34}^{11}$}
& \MainSpan{3}{2}
& \MainSpan{3}{4}
& \MainSpan{3}{$\times3$} \\ \hline

\MainPad{26}
& \MainPad{$|\Psi^-_{12}\rangle|\Psi^-_{34}\rangle-\xi|\Psi^-_{13}\rangle|\Psi^-_{24}\rangle^{\rm e}$}
& \MainPad{$\sum_p\sigma_p^z$}
& \MainPad{$\sum_p\sigma_p^\pm$}
& \MainPad{$\mathfrak{sl}_2$}
& \MainPad{3}
& \MainPad{---}
& \MainPad{0}
& \MainPad{3}
& \MainPad{$\times1$} \\ \hline

\MainPad{27}
& \MainPad{$\tau_1|\Psi^+_{12,34}\rangle+\tau_2|\Psi^+_{13,24}\rangle+|\Psi^+_{14,23}\rangle^{\rm f}$}
& \MainPad{$\sum_p\sigma_p^z$}
& \MainPad{---}
& \MainPad{$\mathbb C_{\rm tor}$}
& \MainPad{1}
& \MainPad{---}
& \MainPad{0}
& \MainPad{1}
& \MainPad{$\times1$} \\ \hline

\MainPad{28}
& \MainPad{$|0000\rangle+|\Psi^-_{12}\rangle|\Psi^-_{34}\rangle-\xi|\Psi^-_{13}\rangle|\Psi^-_{24}\rangle^{\rm e}$}
& \MainPad{---}
& \MainPad{$\sum_p\sigma_p^+$}
& \MainPad{$\mathbb C_{\rm nil}$}
& \MainPad{1}
& \MainPad{---}
& \MainPad{0}
& \MainPad{1}
& \MainPad{$\times1$} \\ \hline

\noalign{\MainGroupHeight
  {$|{\rm GHZ}_4\rangle+\sum_{i=2}^{4}\tau_{i-1}|\Psi^+_{1i,\overline{1i}}\rangle^{\rm f}$}
  {$|{\rm GHZ}_4\rangle+\tau_1|\Psi^+_{12,34}\rangle+\tau_2|\Psi^+_{13,24}\rangle+\tau_3|0110\rangle^{\rm f}$}
  {$\upsilon(|\Phi^+_{14}\rangle|\Phi^+_{23}\rangle-\xi|\Phi^+_{13}\rangle|\Phi^+_{24}\rangle)+|00\rangle|\Psi^+_{34}\rangle-|\Psi^+_{12}\rangle|11\rangle^{\rm e}$}}
&
\rule[-\dimexpr\MainRowHeight/2-2.5pt\relax]{0pt}{\MainRowHeight}
\MainPad{$|{\rm GHZ}_4\rangle+\sum_{i=2}^{4}\tau_{i-1}|\Psi^+_{1i,\overline{1i}}\rangle^{\rm f}$}
& & & & & & & & \\
\MainPartialRule

&
\rule[-\dimexpr\MainRowHeight/2-2.5pt\relax]{0pt}{\MainRowHeight}
\MainPad{$|{\rm GHZ}_4\rangle+\tau_1|\Psi^+_{12,34}\rangle+\tau_2|\Psi^+_{13,24}\rangle+\tau_3|0110\rangle^{\rm f}$}
& & & & & & & & \\
\MainPartialRule

\MainSpan{3}{29}
& \rule[-\dimexpr\MainRowHeight/2-2.5pt\relax]{0pt}{\MainRowHeight}
  \MainPad{$\upsilon(|\Phi^+_{14}\rangle|\Phi^+_{23}\rangle-\xi|\Phi^+_{13}\rangle|\Phi^+_{24}\rangle)+|00\rangle|\Psi^+_{34}\rangle-|\Psi^+_{12}\rangle|11\rangle^{\rm e}$}
& \MainSpan{3}{---}
& \MainSpan{3}{---}
& \MainSpan{3}{$0$}
& \MainSpan{3}{0}
& \MainSpan{3}{---}
& \MainSpan{3}{0}
& \MainSpan{3}{0}
& \MainSpan{3}{$\times1$} \\ \hline

\multicolumn{10}{|>{\columncolor{violet!5}[\tabcolsep][\tabcolsep]}c|}{%
\makebox[0pt]{\MainPad{\parbox{660pt}{\centering\footnotesize
All displayed parameters are nonzero. We define
$\Delta({\bm x})=
\prod_i(x_i^2-1)
\prod_{i<j}(x_i^2-x_j^2)
\prod_{{\bm s}\in\{\pm1\}^n}
(1-\sum_i s_i x_i)$.\\[2pt]
\textsuperscript{a} $\alpha^2\neq1,9$.\quad
\textsuperscript{b} $\delta^2\neq1$.\quad
\textsuperscript{c} $\epsilon^2\neq1,4$.\quad
\textsuperscript{d} $\theta^2\neq1,\delta^2,(1\pm2\delta)^2$.\quad
\textsuperscript{e} $\xi\neq\pm1,\tfrac12,2$.\quad
\textsuperscript{f} $\Delta({\bm\tau})\neq0$.
}}}} \\
\noalign{\global\arrayrulewidth=.8pt}
\hline
\noalign{\global\arrayrulewidth=.4pt}
\end{tabular}
}

\caption{\textbf{Complete algebraic and rank CO classification through $4$ qubits.}
Representative states for the $n=1,2,3,4$ qubit families, their algebraic COs
(Cartan and non-Cartan), and their rank COs in the fermionic representation.
The notation $Q(\lambda)$ for a Cartan operator $Q$ indicates the eigenvalue when $\neq 0$.
Here, $D_{\pstab}$ ($\Delta D$) denote the number of independent algebraic (rank) COs, with $D=D_{\pstab} + \Delta D$ the total number of non-trivial one-body COs in the fermionic representation. }
\label{tabgen}
\end{table*}

The first columns display the algebraic CO generators, separating Cartan and non-Cartan COs, followed by the corresponding Lie algebra and its dimension $D_{\pstab}$. A parenthesis $(\lambda)$ denotes a nonzero CO eigenvalue $\lambda$. The next columns display the rank COs through their equivalent fermionic generators $G_{pq}^{ij}$, introduced in Eq. \eqref{eq:transfer-methods}. When several pairs are written before a colon, the displayed span applies to each of them. Only pairs with $p<q$ are listed, since the reverse directions contain no independent information. Here, $\Delta D$ counts all intersector fermionic COs, including both $pq$ and $qp$ directions. 
The total number of nontrivial one-body fermionic COs is $D=D_{\pstab}+\Delta D$, excluding the $n$ sector-number operators. The column ``deg'' indicates the multiplicity under qubit permutations. Finally, the first column provides representative states for each CO type. We use
\begin{subequations}
\begin{eqnarray}
|\Phi_{\mu,\nu}^{\pm}\rangle
&=& \frac{1}{\sqrt{2}}
\left(|0_{\mu}0_{\nu}\rangle
\pm |1_{\mu}1_{\nu}\rangle\right),\\
|\Psi_{\mu,\nu}^{\pm}\rangle
&=& \frac{1}{\sqrt{2}}
\left(|0_{\mu}1_{\nu}\rangle
\pm |1_{\mu}0_{\nu}\rangle\right),\\
\ket{W_n}
&=& \tfrac{1}{\sqrt{n}}
\sum_{i=1}^n \ket{0\cdots 1_i\cdots 0},\\
\ket{\mathrm{GHZ}_n}
&=& \tfrac{1}{\sqrt{2}}
(\ket0^{\otimes n}+\ket1^{\otimes n}),
\end{eqnarray}
\end{subequations}
where $\mu,\nu$ denote either individual qubits or collective qubit blocks (and $\ket{i_\mu} = \bigotimes_{p \in \mu} |i_p\rangle$). Some CO types contain multiple SLOCC branches, for which different representative state families are included. This ensures that every state belonging to a given CO type is SLOCC-equivalent, up to qubit permutations, to at least one of the listed representatives for some allowed choice of parameters.

To finish, we note thar in the table a notation such as $\sigma_{1,2}^+-\sigma_{2,3}^+$ indicates that both $\sigma_{1}^+-\sigma_{2}^+$ and $\sigma_{2}^+-\sigma_{3}^+$ are COs.

\section*{Code availability}
Code can be made available upon reasonable requests.

\section*{Data availability}
The complete one-through-four-qubit tables are in the Methods,  the complete five-qubit tables are in the Supplementary Information, and the six qubit table is presented in a Supplemental Data file. 

\section*{Acknowledgements}
FP and MC were supported by the Laboratory Directed Research and Development (LDRD) program of Los Alamos National Laboratory (LANL) under project number 20260043DR. This material is based upon work supported by the U.S. Department of Energy, Office of Science, National Quantum Information Science Research Centers, Quantum Science Center (LC).

\section*{Author Contributions}
The project was conceived by FP and RR. MC led the algebraic calculations. All authors contributed to the theory, writing and reviewing of the manuscript. 

\section*{Competing Interests statements}
The authors declare no competing interests.

\section*{Artificial Intelligence Disclosure}
The idea for this project was created by humans. The authors acknowledge the use of ChatGPT Astra for the deterministic code that performs the full classification.

\clearpage


%

\clearpage
\onecolumngrid
\appendix

\begin{center}
	{\Large \bf Supplementary Information for ``Classification of quantum states through conserved operators''} 
\end{center}

This Supplementary Information provides additional details supporting the results presented in the main text. We first prove the bounds on the number of algebraic COs and their saturation by highest-weight states. We then derive the adapted form of local qubit CO algebras used in the classification, followed by an independent four-qubit cross-check against known SLOCC stabilizer families. Finally, we present the complete five-qubit algebraic-plus-rank CO classification.

\section{Bounds on algebraic COs and highest-weight states}
\label{sec:highest-weight}

\begin{theorem}[Bounds on algebraic COs and orbit dimensions]
Let $G$ be a connected complex reductive group acting irreducibly on $\mathcal H$, with Lie algebra $\mathfrak g$. For every nonzero $\ket{\psi}\in\mathcal H$,
\begin{equation}
\max\{0,D_{\mathfrak g}+1-D_{\mathcal H}\}
\leq D_{\pstab[\psi]}
\leq D_{\pstab[\psi_{\rm hw}]}.
\label{eq:si-bounds}
\end{equation}
The upper bound on the number of algebraic COs, and equivalently the lower bound on the orbit dimension, is saturated if and only if $[\psi]\in G\cdot[\psi_{\rm hw}]$.
\end{theorem}

\subsection{The lower bound}

Define the augmented action map
\begin{equation}
A_\psi:\mathfrak g\oplus\C\longrightarrow\mathcal H,
\qquad
A_\psi(Q,a)=Q\ket\psi-a\ket\psi.
\label{eq:si-augmented-map}
\end{equation}
For nonzero $\ket\psi$, the projection
\begin{equation}
\ker A_\psi\longrightarrow\pstab[\psi],
\qquad (Q,a)\longmapsto Q,
\end{equation}
is an isomorphism, since the scalar $a$ is uniquely fixed to $\chi_\psi(Q)$. Hence,
\begin{equation}
D_{\pstab[\psi]}=D_{\mathfrak g}+1-\rank A_\psi.
\label{eq:si-augmented-rank}
\end{equation}
Since $\rank A_\psi\leq D_{\mathcal H}$ and $D_{\pstab[\psi]}\geq 0$, we immediately obtain the lower bound in Eq.~\eqref{eq:si-bounds}.

\subsection{The upper bound and saturation}

Let $V_\lambda$ be an irreducible finite-dimensional representation of $G$, with highest-weight vector $\ket{\psi_{\rm hw}}$. Choose a Borel subgroup $B\subset G$ compatible with this highest-weight convention. A ray fixed by $B$ must be annihilated by every positive-root operator, since these operators act nilpotently and therefore have no nonzero eigenvalues. In an irreducible highest-weight representation, the common kernel of all positive-root operators is the highest-weight line. Thus, $[\psi_{\rm hw}]$ is the unique $B$-fixed ray.

For any nonzero $\ket\psi\in V_\lambda$, consider the Zariski closure $\overline{G\cdot[\psi]}$. By the Borel fixed-point theorem~\cite{humphreys1975linear}, it contains a $B$-fixed ray, hence $[\psi_{\rm hw}]$. Since the closure is $G$-invariant,
\begin{equation}
G\cdot[\psi_{\rm hw}]
\subseteq \overline{G\cdot[\psi]}.
\label{eq:si-hw-closure}
\end{equation}
As an orbit and its Zariski closure have the same dimension,
\begin{equation}
\dim G\cdot[\psi_{\rm hw}]
\leq \dim G\cdot[\psi].
\end{equation}
Using
\begin{equation}
\dim G\cdot[\psi]=D_{\mathfrak g}-D_{\pstab[\psi]},
\end{equation}
we obtain the upper bound in Eq.~\eqref{eq:si-bounds}.

To prove the saturation condition, note that the highest-weight orbit is closed, since its stabilizer is parabolic and $G\cdot[\psi_{\rm hw}]\simeq G/P_\lambda$ is projective. If the upper bound is saturated, the two orbits have the same dimension. By Eq.~\eqref{eq:si-hw-closure}, the closed irreducible variety $G\cdot[\psi_{\rm hw}]$ is contained in the irreducible variety $\overline{G\cdot[\psi]}$ and has the same dimension. Since a proper closed subvariety of an irreducible variety has strictly smaller dimension, we conclude that
\begin{equation}
\overline{G\cdot[\psi]}=G\cdot[\psi_{\rm hw}],
\end{equation}
which implies $G\cdot[\psi]=G\cdot[\psi_{\rm hw}]$. The converse follows because stabilizer dimensions are constant along an orbit.

For local qudit transformations, the highest-weight orbit is precisely the set of product states, yielding the upper bound quoted in the main text.

\section{Adapted form of local qubit CO algebras}
\label{sec:adapted}

Here we provide additional details on how the CO algebras
are classified. In particular, recall that the algebra
associated with SLOCC is
\begin{equation}
\mathfrak g^{\rm dis}
=
\bigoplus_{p=1}^n\mathfrak{sl}_{2,p}(\C),
\end{equation}
and that every CO algebra $\pstab[\psi]$ is a
subalgebra of $\mathfrak g^{\rm dis}$.

\begin{proposition}[Adapted form]
Up to local conjugation and qubit permutation, every
complete local qubit CO algebra can be written as
\begin{equation}
\pstab[\psi]
=
\bigoplus_{k=1}^K \pstab_k[\psi]
\oplus \bar{\pstab}[\psi],
\label{eq:si-adapted-form}
\end{equation}
where
\begin{equation}
\pstab_k[\psi] = \Span\left\{\sigma_{\mathcal P_k}^z,\sigma_{\mathcal P_k}^+,\sigma_{\mathcal P_k}^- \right\}, \quad \sigma_{\mathcal P_k}^{\mu} = \sum_{p\in\mathcal P_k}\sigma_p^\mu
\end{equation}
and
\begin{equation}
\bar{\pstab}[\psi] = \bar{\pstab}^z[\psi]
\ltimes \bar{\pstab}^+[\psi], \quad \bar{\pstab}^{\mu}[\psi]\subseteq\bar{\mathfrak g}^{\mu}=\Span\{\sigma_{\bar p}^{\mu}:
\bar p\in\bar{\mathcal P}\}.
\label{eq:si-solvable-form}
\end{equation}
Here, $\mathcal P_1,\ldots,\mathcal P_K$ are disjoint sets of sites, with $\bar{\mathcal P}$ containing the remaining sites. Notice that $\bar{\pstab}^z[\psi]$ is toral, while
$\bar{\pstab}^+[\psi]$ is an abelian ideal stable under the adjoint action of $\bar{\pstab}^z[\psi]$.
\end{proposition}

\begin{proof}
Take a Levi decomposition
\begin{equation}
\pstab[\psi]
=
\mathfrak s\ltimes\rad(\pstab[\psi]),
\end{equation}
with $\mathfrak s$ semisimple.
Every simple factor of $\mathfrak s$ projects to each
local $\mathfrak{sl}_2$ either trivially or injectively.
A nonzero projection is therefore an isomorphism onto
$\mathfrak{sl}_2$.
Two distinct simple factors cannot both project
nontrivially onto the same local site, since their
images would be commuting nonzero ideals inside
the simple algebra $\mathfrak{sl}_2$.
Hence the simple factors have disjoint supports.

For a simple factor supported on a block $\mathcal P_k$,
every nonzero local projection differs from a reference
copy of $\mathfrak{sl}_2$ by an automorphism.
Moreover, all automorphisms of $\mathfrak{sl}_2(\C)$
are inner, so independent local conjugations on the
sites in $\mathcal P_k$ bring the factor to the
collective form $\pstab_k[\psi]$.
The projection of the solvable radical onto any
site in $\mathcal P_k$ is an $\mathfrak{sl}_2$-stable
solvable ideal of the local $\mathfrak{sl}_2$,
and therefore vanishes.
Thus the solvable radical has no support on those
sites, which gives the direct sum in
Eq.~\eqref{eq:si-adapted-form}.

On the remaining sites $\bar{\mathcal P}$,
the algebra $\bar{\pstab}[\psi]$ is solvable.
Moreover, since $\pstab[\psi]$ is the Lie algebra
of the projective stabilizer of $[\psi]$, it is
an algebraic Lie algebra.
Hence, the corresponding connected solvable
algebraic group admits a semidirect decomposition
into a torus and its unipotent radical.

This solvable group can be conjugated into a product
of local Borel subgroups. By local conjugations,
the torus can be chosen inside the diagonal Cartan
subalgebra $\bar{\mathfrak g}^z$, with Lie algebra $\bar{\pstab}^z[\psi]$. The Lie algebra of the unipotent radical is then
a $T$-stable subspace of $\bar{\mathfrak g}^+$,
and hence coincides with $\bar{\pstab}^+[\psi]$. Since all generators of $\bar{\mathfrak g}^+$ commute, this subspace is abelian. Together with the Levi factors above, this completes the decomposition in
Eq.~\eqref{eq:si-adapted-form}.
\end{proof}

\section{Independent four-qubit cross-check}
\label{sec:four-qubit-check}

The algebraic part of the four-qubit result can be checked independently against the complete SLOCC stabilizer stratification of Ref.~\cite{dietrich2022classification}. Their classification is expressed in terms of the vector stabilizer
\begin{equation}
\mathfrak z [\psi] = \{Q\in\mathfrak g:Q\ket\psi=0\},
\end{equation}
whereas we use the projective stabilizer $\pstab{[\psi]}$. Since
\begin{equation}
\ker\chi_\psi=\mathfrak z[\psi],
\end{equation}
we have
\begin{equation}
\dim \pstab{[\psi]} - \dim\mathfrak z[\psi] = \dim\operatorname{Im}\chi_\psi \in\{0,1\}.
\end{equation}
For semisimple and mixed strata admitting a nonzero homogeneous invariant, the projective character vanishes. For a nonzero nilpotent orbit, the normal $\mathfrak{sl}_2$ triple supplies a semisimple element acting nontrivially on the ray, giving the possible extra projective direction.

Evaluating the augmented local action matrix on exact representatives of the known four-qubit strata and quotienting by qubit permutations gives exactly $17$ algebraic CO types. Adding the pair spaces and repeating the simultaneous quotient splits these into the $29$ algebraic-plus-rank types listed in the Methods of the main text. The agreement gives an external check of the four-qubit enumeration that does not use the five- and six-qubit calculation.

\section{Supplementary classification data}
\label{sec:data-format}

The complete $5$-qubit algebraic-plus-rank classification contains $287$ fixed types and $4$ continuous families. Table \ref{tab:si-n5-complete} gives one representative state together with the complete algebraic and rank CO data for every type. 

The parameter conditions for the continuous families are indicated by superscripts \textsuperscript{a} and \textsuperscript{b} in the table and stated in the shaded row at the end of the table.

We use the normalized Bell states, the $|W_n\rangle$ and the $|{\rm GHZ}_n\rangle$ of the main text. Overall normalization of the representatives is immaterial. \\

\begingroup
\newcommand{\SIpagedrow}[1]{%
  \raisebox{\depth}{\resizebox{0.99\linewidth}{!}{%
    \setlength{\tabcolsep}{1.3pt}%
    \renewcommand{\arraystretch}{1.18}%
    \scriptsize
    \begin{tabular}{|>{\columncolor{gray!15}[\tabcolsep][\tabcolsep]\centering\arraybackslash}m{18pt}|>{\columncolor{white}[\tabcolsep][\tabcolsep]\centering\arraybackslash}m{174pt}|>{\columncolor{red!10}[\tabcolsep][\tabcolsep]\centering\arraybackslash}m{76pt}|>{\columncolor{red!10}[\tabcolsep][\tabcolsep]\centering\arraybackslash}m{78pt}|>{\columncolor{red!10}[\tabcolsep][\tabcolsep]\centering\arraybackslash}m{62pt}|>{\columncolor{red!10}[\tabcolsep][\tabcolsep]\centering\arraybackslash}m{16pt}|>{\columncolor{yellow!15}[\tabcolsep][\tabcolsep]\centering\arraybackslash}m{145pt}|>{\columncolor{yellow!15}[\tabcolsep][\tabcolsep]\centering\arraybackslash}m{18pt}|>{\columncolor{green!10}[\tabcolsep][\tabcolsep]\centering\arraybackslash}m{16pt}|>{\columncolor{green!10}[\tabcolsep][\tabcolsep]\centering\arraybackslash}m{20pt}|}
    #1
    \end{tabular}%
  }}%
}
\setlength{\LTleft}{\fill}
\setlength{\LTright}{\fill}
\setlength{\LTpre}{0pt}
\setlength{\LTpost}{0pt}
\setlength{\LTcapwidth}{0.99\linewidth}
\setlength{\tabcolsep}{0pt}
\renewcommand{\arraystretch}{0}
\begin{longtable}{@{}c@{}}
\SIpagedrow{\hhline{~|---------|}
\multicolumn{1}{p{18pt}|}{\rule[-4pt]{0pt}{14pt}} & & \multicolumn{4}{>{\columncolor{red!15}[\tabcolsep][\tabcolsep]}c|}{\bfseries\footnotesize Algebraic/Intrasector COs} & \multicolumn{2}{>{\columncolor{yellow!20}[\tabcolsep][\tabcolsep]}c|}{\bfseries\footnotesize Rank/Intersector COs} & \cellcolor{green!15} & \cellcolor{green!15} \\
\hhline{~|>{\arrayrulecolor{white}}->{\arrayrulecolor{black}}|----|--|>{\arrayrulecolor{green!15}}->{\arrayrulecolor{black}}|>{\arrayrulecolor{green!15}}->{\arrayrulecolor{black}}|}
\multicolumn{1}{p{18pt}|}{\rule[-4pt]{0pt}{14pt}} & \multirow{-2}{*}[3pt]{\bfseries\footnotesize Representative state} & \cellcolor{red!15}\bfseries\footnotesize Cartan & \cellcolor{red!15}\bfseries\footnotesize Non-Cartan & \cellcolor{red!15}\bfseries\footnotesize Algebra & \cellcolor{red!15}$D_{\pstab}$ & \cellcolor{yellow!20}\bfseries\footnotesize COs ($p<q$) & \cellcolor{yellow!20}$\tfrac{\Delta D}{2}$ & \cellcolor{green!15}\multirow{-2}{*}[3pt]{$D$} & \cellcolor{green!15}\multirow{-2}{*}[3pt]{\bfseries\footnotesize deg} \\
\hline} \\
\endfirsthead
\multicolumn{1}{c}{\footnotesize Table \thetable\ (continued)} \\[3pt]
\SIpagedrow{\hhline{~|---------|}
\multicolumn{1}{p{18pt}|}{\rule[-4pt]{0pt}{14pt}} & & \multicolumn{4}{>{\columncolor{red!15}[\tabcolsep][\tabcolsep]}c|}{\bfseries\footnotesize Algebraic/Intrasector COs} & \multicolumn{2}{>{\columncolor{yellow!20}[\tabcolsep][\tabcolsep]}c|}{\bfseries\footnotesize Rank/Intersector COs} & \cellcolor{green!15} & \cellcolor{green!15} \\
\hhline{~|>{\arrayrulecolor{white}}->{\arrayrulecolor{black}}|----|--|>{\arrayrulecolor{green!15}}->{\arrayrulecolor{black}}|>{\arrayrulecolor{green!15}}->{\arrayrulecolor{black}}|}
\multicolumn{1}{p{18pt}|}{\rule[-4pt]{0pt}{14pt}} & \multirow{-2}{*}[3pt]{\bfseries\footnotesize Representative state} & \cellcolor{red!15}\bfseries\footnotesize Cartan & \cellcolor{red!15}\bfseries\footnotesize Non-Cartan & \cellcolor{red!15}\bfseries\footnotesize Algebra & \cellcolor{red!15}$D_{\pstab}$ & \cellcolor{yellow!20}\bfseries\footnotesize COs ($p<q$) & \cellcolor{yellow!20}$\tfrac{\Delta D}{2}$ & \cellcolor{green!15}\multirow{-2}{*}[3pt]{$D$} & \cellcolor{green!15}\multirow{-2}{*}[3pt]{\bfseries\footnotesize deg} \\
\hline} \\
\endhead
\endfoot
\endlastfoot
\SIpagedrow{1 & \SIcell{174pt}{\ket{00000}} & \SIcell{76pt}{\sigma_{p}^{z}(1)} & \SIcell{78pt}{\sigma_{p}^{+}} & \SIcell{62pt}{\mathfrak b^{\oplus5}} & 10 & \SIcell{145pt}{G_{pq}^{0i},\,G_{pq}^{11}} & 30 & 70 & $\times1$ \\ \hline} \\
\SIpagedrow{2 & \SIcell{174pt}{\ket{000}\ket{\Psi^-_{45}}} & \SIcell{76pt}{\sigma_{1,2,3}^{z}(1),\,\sigma_{4}^{z}+\sigma_{5}^{z}} & \SIcell{78pt}{\sigma_{1,2,3}^{+},\,\sigma_{4}^{\pm}+\sigma_{5}^{\pm}} & \SIcell{62pt}{\mathfrak b^{\oplus3}\oplus\mathfrak{sl}_2} & 9 & \SIcell{145pt}{\begin{gathered}G_{pq}^{0i},\,G_{pq}^{11}\;(p,q=1,2,3),\,G_{p4}^{0i}\\[-2pt]G_{p5}^{0i}\;(p=1,2,3),\,G_{45}^{00}-G_{45}^{11},\,G_{45}^{01},\,G_{45}^{10}\end{gathered}} & 24 & 57 & $\times10$ \\ \hline} \\
\SIpagedrow{3 & \SIcell{174pt}{\ket{00}\ket{W_3}} & \SIcell{76pt}{\sigma_{1,2}^{z}(1),\,\sum_{p=3}^{5}\sigma_p^{z}(1)} & \SIcell{78pt}{\sigma_{1,2}^{+},\,\sigma_{3,4}^{+}-\sigma_{4,5}^{+}} & \SIcell{62pt}{\mathfrak b^{\oplus2}\oplus(\C\ltimes\C^2)} & 7 & \SIcell{145pt}{\begin{gathered}G_{12}^{0i},\,G_{12}^{11},\,G_{1q}^{0i}\;(q=3,4,5),\,G_{2q}^{0i}\\[-2pt]G_{pq}^{00}+G_{pq}^{11},\,G_{pq}^{01}\;(p,q=3,4,5)\end{gathered}} & 21 & 49 & $\times10$ \\ \hline} \\
\SIpagedrow{4 & \SIcell{174pt}{\ket{00}\ket{\mathrm{GHZ}_3}} & \SIcell{76pt}{\sigma_{1,2}^{z}(1),\,\sigma_{3,4}^{z}-\sigma_{4,5}^{z}} & \SIcell{78pt}{\sigma_{1,2}^{+}} & \SIcell{62pt}{\mathfrak b^{\oplus2}\oplus\C_{\rm tor}^{2}} & 6 & \SIcell{145pt}{\begin{gathered}G_{12}^{0i},\,G_{12}^{11},\,G_{1q}^{0i}\;(q=3,4,5),\,G_{2q}^{0i}\\[-2pt]G_{pq}^{ii}\;(p,q=3,4,5)\end{gathered}} & 21 & 48 & $\times10$ \\ \hline} \\
\SIpagedrow{5 & \SIcell{174pt}{\ket{0}\ket{W_4}} & \SIcell{76pt}{\sigma_{1}^{z}(1),\,\sum_{p=2}^{5}\sigma_p^{z}(2)} & \SIcell{78pt}{\sigma_{1}^{+},\,\sigma_{2,3,4}^{+}-\sigma_{3,4,5}^{+}} & \SIcell{62pt}{\mathfrak b\oplus(\C\ltimes\C^3)} & 6 & \SIcell{145pt}{G_{1q}^{0i},\,G_{pq}^{00}+G_{pq}^{11},\,G_{pq}^{01}\;(p,q=2,3,4,5)} & 20 & 46 & $\times5$ \\ \hline} \\
\end{longtable}
\endgroup

\begin{thebibliography}{42}%
\makeatletter
\providecommand \@ifxundefined [1]{%
 \@ifx{#1\undefined}
}%
\providecommand \@ifnum [1]{%
 \ifnum #1\expandafter \@firstoftwo
 \else \expandafter \@secondoftwo
 \fi
}%
\providecommand \@ifx [1]{%
 \ifx #1\expandafter \@firstoftwo
 \else \expandafter \@secondoftwo
 \fi
}%
\providecommand \natexlab [1]{#1}%
\providecommand \enquote  [1]{``#1''}%
\providecommand \bibnamefont  [1]{#1}%
\providecommand \bibfnamefont [1]{#1}%
\providecommand \citenamefont [1]{#1}%
\providecommand \href@noop [0]{\@secondoftwo}%
\providecommand \href [0]{\begingroup \@sanitize@url \@href}%
\providecommand \@href[1]{\@@startlink{#1}\@@href}%
\providecommand \@@href[1]{\endgroup#1\@@endlink}%
\providecommand \@sanitize@url [0]{\catcode `\\12\catcode `\$12\catcode `\&12\catcode `\#12\catcode `\^12\catcode `\_12\catcode `\%12\relax}%
\providecommand \@@startlink[1]{}%
\providecommand \@@endlink[0]{}%
\providecommand \url  [0]{\begingroup\@sanitize@url \@url }%
\providecommand \@url [1]{\endgroup\@href {#1}{\urlprefix }}%
\providecommand \urlprefix  [0]{URL }%
\providecommand \Eprint [0]{\href }%
\providecommand \doibase [0]{https://doi.org/}%
\providecommand \selectlanguage [0]{\@gobble}%
\providecommand \bibinfo  [0]{\@secondoftwo}%
\providecommand \bibfield  [0]{\@secondoftwo}%
\providecommand \translation [1]{[#1]}%
\providecommand \BibitemOpen [0]{}%
\providecommand \bibitemStop [0]{}%
\providecommand \bibitemNoStop [0]{.\EOS\space}%
\providecommand \EOS [0]{\spacefactor3000\relax}%
\providecommand \BibitemShut  [1]{\csname bibitem#1\endcsname}%
\let\auto@bib@innerbib\@empty
\bibitem [{\citenamefont {Chitambar}\ and\ \citenamefont {Gour}(2019)}]{chitambar2019quantum}%
  \BibitemOpen
  \bibfield  {author} {\bibinfo {author} {\bibfnamefont {E.}~\bibnamefont {Chitambar}}\ and\ \bibinfo {author} {\bibfnamefont {G.}~\bibnamefont {Gour}},\ }\bibfield  {title} {\bibinfo {title} {Quantum resource theories},\ }\href {https://doi.org/10.1103/RevModPhys.91.025001} {\bibfield  {journal} {\bibinfo  {journal} {Reviews of modern physics}\ }\textbf {\bibinfo {volume} {91}},\ \bibinfo {pages} {025001} (\bibinfo {year} {2019})}\BibitemShut {NoStop}%
\bibitem [{\citenamefont {Chitambar}\ \emph {et~al.}(2014)\citenamefont {Chitambar}, \citenamefont {Leung}, \citenamefont {Man{\v{c}}inska}, \citenamefont {Ozols},\ and\ \citenamefont {Winter}}]{chitambar2014everything}%
  \BibitemOpen
  \bibfield  {author} {\bibinfo {author} {\bibfnamefont {E.}~\bibnamefont {Chitambar}}, \bibinfo {author} {\bibfnamefont {D.}~\bibnamefont {Leung}}, \bibinfo {author} {\bibfnamefont {L.}~\bibnamefont {Man{\v{c}}inska}}, \bibinfo {author} {\bibfnamefont {M.}~\bibnamefont {Ozols}},\ and\ \bibinfo {author} {\bibfnamefont {A.}~\bibnamefont {Winter}},\ }\bibfield  {title} {\bibinfo {title} {Everything you always wanted to know about locc (but were afraid to ask)},\ }\href {https://doi.org/10.1007/s00220-014-1953-9} {\bibfield  {journal} {\bibinfo  {journal} {Communications in Mathematical Physics}\ }\textbf {\bibinfo {volume} {328}},\ \bibinfo {pages} {303} (\bibinfo {year} {2014})}\BibitemShut {NoStop}%
\bibitem [{\citenamefont {D{\"u}r}\ \emph {et~al.}(2000)\citenamefont {D{\"u}r}, \citenamefont {Vidal},\ and\ \citenamefont {Cirac}}]{dur2000three}%
  \BibitemOpen
  \bibfield  {author} {\bibinfo {author} {\bibfnamefont {W.}~\bibnamefont {D{\"u}r}}, \bibinfo {author} {\bibfnamefont {G.}~\bibnamefont {Vidal}},\ and\ \bibinfo {author} {\bibfnamefont {J.~I.}\ \bibnamefont {Cirac}},\ }\bibfield  {title} {\bibinfo {title} {Three qubits can be entangled in two inequivalent ways},\ }\href {https://doi.org/10.1103/PhysRevA.62.062314} {\bibfield  {journal} {\bibinfo  {journal} {Physical Review A}\ }\textbf {\bibinfo {volume} {62}},\ \bibinfo {pages} {062314} (\bibinfo {year} {2000})}\BibitemShut {NoStop}%
\bibitem [{\citenamefont {Miyake}(2003)}]{miyake2003classification}%
  \BibitemOpen
  \bibfield  {author} {\bibinfo {author} {\bibfnamefont {A.}~\bibnamefont {Miyake}},\ }\bibfield  {title} {\bibinfo {title} {Classification of multipartite entangled states by multidimensional determinants},\ }\href {https://doi.org/10.1103/PhysRevA.67.012108} {\bibfield  {journal} {\bibinfo  {journal} {Physical Review A}\ }\textbf {\bibinfo {volume} {67}},\ \bibinfo {pages} {012108} (\bibinfo {year} {2003})}\BibitemShut {NoStop}%
\bibitem [{\citenamefont {Barnum}\ \emph {et~al.}(2003)\citenamefont {Barnum}, \citenamefont {Knill}, \citenamefont {Ortiz},\ and\ \citenamefont {Viola}}]{barnum2003generalizations}%
  \BibitemOpen
  \bibfield  {author} {\bibinfo {author} {\bibfnamefont {H.}~\bibnamefont {Barnum}}, \bibinfo {author} {\bibfnamefont {E.}~\bibnamefont {Knill}}, \bibinfo {author} {\bibfnamefont {G.}~\bibnamefont {Ortiz}},\ and\ \bibinfo {author} {\bibfnamefont {L.}~\bibnamefont {Viola}},\ }\bibfield  {title} {\bibinfo {title} {Generalizations of entanglement based on coherent states and convex sets},\ }\href {https://doi.org/10.1103/PhysRevA.68.032308} {\bibfield  {journal} {\bibinfo  {journal} {Physical Review A}\ }\textbf {\bibinfo {volume} {68}},\ \bibinfo {pages} {032308} (\bibinfo {year} {2003})}\BibitemShut {NoStop}%
\bibitem [{\citenamefont {Diaz}\ \emph {et~al.}(2025)\citenamefont {Diaz}, \citenamefont {Mele}, \citenamefont {Bermejo}, \citenamefont {Braccia}, \citenamefont {Larocca},\ and\ \citenamefont {Cerezo}}]{diaz2025unified}%
  \BibitemOpen
  \bibfield  {author} {\bibinfo {author} {\bibfnamefont {N.~L.}\ \bibnamefont {Diaz}}, \bibinfo {author} {\bibfnamefont {A.~A.}\ \bibnamefont {Mele}}, \bibinfo {author} {\bibfnamefont {P.}~\bibnamefont {Bermejo}}, \bibinfo {author} {\bibfnamefont {P.}~\bibnamefont {Braccia}}, \bibinfo {author} {\bibfnamefont {M.}~\bibnamefont {Larocca}},\ and\ \bibinfo {author} {\bibfnamefont {M.}~\bibnamefont {Cerezo}},\ }\bibfield  {title} {\bibinfo {title} {A unified approach to quantum resource theories and a new class of free operations},\ }\bibfield  {journal} {\bibinfo  {journal} {arXiv preprint arXiv:2507.10851}\ }\href {https://doi.org/10.48550/arXiv.2507.10851} {10.48550/arXiv.2507.10851} (\bibinfo {year} {2025})\BibitemShut {NoStop}%
\bibitem [{\citenamefont {Verstraete}\ \emph {et~al.}(2002)\citenamefont {Verstraete}, \citenamefont {Dehaene}, \citenamefont {De~Moor},\ and\ \citenamefont {Verschelde}}]{verstraete2002four}%
  \BibitemOpen
  \bibfield  {author} {\bibinfo {author} {\bibfnamefont {F.}~\bibnamefont {Verstraete}}, \bibinfo {author} {\bibfnamefont {J.}~\bibnamefont {Dehaene}}, \bibinfo {author} {\bibfnamefont {B.}~\bibnamefont {De~Moor}},\ and\ \bibinfo {author} {\bibfnamefont {H.}~\bibnamefont {Verschelde}},\ }\bibfield  {title} {\bibinfo {title} {Four qubits can be entangled in nine different ways},\ }\href {https://doi.org/10.1103/PhysRevA.65.052112} {\bibfield  {journal} {\bibinfo  {journal} {Physical Review A}\ }\textbf {\bibinfo {volume} {65}},\ \bibinfo {pages} {052112} (\bibinfo {year} {2002})}\BibitemShut {NoStop}%
\bibitem [{\citenamefont {Dietrich}\ \emph {et~al.}(2022)\citenamefont {Dietrich}, \citenamefont {De~Graaf}, \citenamefont {Marrani},\ and\ \citenamefont {Origlia}}]{dietrich2022classification}%
  \BibitemOpen
  \bibfield  {author} {\bibinfo {author} {\bibfnamefont {H.}~\bibnamefont {Dietrich}}, \bibinfo {author} {\bibfnamefont {W.~A.}\ \bibnamefont {De~Graaf}}, \bibinfo {author} {\bibfnamefont {A.}~\bibnamefont {Marrani}},\ and\ \bibinfo {author} {\bibfnamefont {M.}~\bibnamefont {Origlia}},\ }\bibfield  {title} {\bibinfo {title} {Classification of four qubit states and their stabilisers under slocc operations},\ }\href {https://doi.org/10.1088/1751-8121/ac4b13} {\bibfield  {journal} {\bibinfo  {journal} {Journal of Physics A: Mathematical and Theoretical}\ }\textbf {\bibinfo {volume} {55}},\ \bibinfo {pages} {095302} (\bibinfo {year} {2022})}\BibitemShut {NoStop}%
\bibitem [{\citenamefont {Jia}\ \emph {et~al.}(2015)\citenamefont {Jia}, \citenamefont {Li},\ and\ \citenamefont {Qiao}}]{jia2015classification}%
  \BibitemOpen
  \bibfield  {author} {\bibinfo {author} {\bibfnamefont {K.-K.}\ \bibnamefont {Jia}}, \bibinfo {author} {\bibfnamefont {J.-L.}\ \bibnamefont {Li}},\ and\ \bibinfo {author} {\bibfnamefont {C.-F.}\ \bibnamefont {Qiao}},\ }\bibfield  {title} {\bibinfo {title} {Classification of the entangled states of $2\times l\times m\times n\times h$},\ }\bibfield  {journal} {\bibinfo  {journal} {arXiv preprint arXiv:1505.02007}\ }\href {https://doi.org/10.48550/arXiv.1505.02007} {10.48550/arXiv.1505.02007} (\bibinfo {year} {2015})\BibitemShut {NoStop}%
\bibitem [{\citenamefont {Chterental}\ and\ \citenamefont {Djokovic}(2006)}]{chterental2006normal}%
  \BibitemOpen
  \bibfield  {author} {\bibinfo {author} {\bibfnamefont {O.}~\bibnamefont {Chterental}}\ and\ \bibinfo {author} {\bibfnamefont {D.~Z.}\ \bibnamefont {Djokovic}},\ }\bibfield  {title} {\bibinfo {title} {Normal forms and tensor ranks of pure states of four qubits},\ }\bibfield  {journal} {\bibinfo  {journal} {arXiv preprint quant-ph/0612184}\ }\href {https://doi.org/10.48550/arXiv.quant-ph/0612184} {10.48550/arXiv.quant-ph/0612184} (\bibinfo {year} {2006})\BibitemShut {NoStop}%
\bibitem [{\citenamefont {Luque}\ and\ \citenamefont {Thibon}(2003)}]{luque2003polynomial}%
  \BibitemOpen
  \bibfield  {author} {\bibinfo {author} {\bibfnamefont {J.-G.}\ \bibnamefont {Luque}}\ and\ \bibinfo {author} {\bibfnamefont {J.-Y.}\ \bibnamefont {Thibon}},\ }\bibfield  {title} {\bibinfo {title} {Polynomial invariants of four qubits},\ }\href {https://doi.org/10.1103/PhysRevA.67.042303} {\bibfield  {journal} {\bibinfo  {journal} {Physical Review A}\ }\textbf {\bibinfo {volume} {67}},\ \bibinfo {pages} {042303} (\bibinfo {year} {2003})}\BibitemShut {NoStop}%
\bibitem [{\citenamefont {Gour}\ and\ \citenamefont {Wallach}(2013)}]{gour2013classification}%
  \BibitemOpen
  \bibfield  {author} {\bibinfo {author} {\bibfnamefont {G.}~\bibnamefont {Gour}}\ and\ \bibinfo {author} {\bibfnamefont {N.~R.}\ \bibnamefont {Wallach}},\ }\bibfield  {title} {\bibinfo {title} {Classification of multipartite entanglement of all finite dimensionality},\ }\href {https://doi.org/10.1103/PhysRevLett.111.060502} {\bibfield  {journal} {\bibinfo  {journal} {Physical Review Letters}\ }\textbf {\bibinfo {volume} {111}},\ \bibinfo {pages} {060502} (\bibinfo {year} {2013})}\BibitemShut {NoStop}%
\bibitem [{\citenamefont {Walter}\ \emph {et~al.}(2013)\citenamefont {Walter}, \citenamefont {Doran}, \citenamefont {Gross},\ and\ \citenamefont {Christandl}}]{walter2013entanglement}%
  \BibitemOpen
  \bibfield  {author} {\bibinfo {author} {\bibfnamefont {M.}~\bibnamefont {Walter}}, \bibinfo {author} {\bibfnamefont {B.}~\bibnamefont {Doran}}, \bibinfo {author} {\bibfnamefont {D.}~\bibnamefont {Gross}},\ and\ \bibinfo {author} {\bibfnamefont {M.}~\bibnamefont {Christandl}},\ }\bibfield  {title} {\bibinfo {title} {Entanglement polytopes: multiparticle entanglement from single-particle information},\ }\href {https://doi.org/10.1126/science.1232957} {\bibfield  {journal} {\bibinfo  {journal} {Science}\ }\textbf {\bibinfo {volume} {340}},\ \bibinfo {pages} {1205} (\bibinfo {year} {2013})}\BibitemShut {NoStop}%
\bibitem [{\citenamefont {Gharahi}\ \emph {et~al.}(2020)\citenamefont {Gharahi}, \citenamefont {Mancini},\ and\ \citenamefont {Ottaviani}}]{gharahi2020fine}%
  \BibitemOpen
  \bibfield  {author} {\bibinfo {author} {\bibfnamefont {M.}~\bibnamefont {Gharahi}}, \bibinfo {author} {\bibfnamefont {S.}~\bibnamefont {Mancini}},\ and\ \bibinfo {author} {\bibfnamefont {G.}~\bibnamefont {Ottaviani}},\ }\bibfield  {title} {\bibinfo {title} {Fine-structure classification of multiqubit entanglement by algebraic geometry},\ }\href {https://doi.org/10.1103/PhysRevResearch.2.043003} {\bibfield  {journal} {\bibinfo  {journal} {Physical Review Research}\ }\textbf {\bibinfo {volume} {2}},\ \bibinfo {pages} {043003} (\bibinfo {year} {2020})}\BibitemShut {NoStop}%
\bibitem [{\citenamefont {Briand}\ \emph {et~al.}(2003)\citenamefont {Briand}, \citenamefont {Luque},\ and\ \citenamefont {Thibon}}]{briand2003complete}%
  \BibitemOpen
  \bibfield  {author} {\bibinfo {author} {\bibfnamefont {E.}~\bibnamefont {Briand}}, \bibinfo {author} {\bibfnamefont {J.-G.}\ \bibnamefont {Luque}},\ and\ \bibinfo {author} {\bibfnamefont {J.-Y.}\ \bibnamefont {Thibon}},\ }\bibfield  {title} {\bibinfo {title} {A complete set of covariants of the four qubit system},\ }\href {https://doi.org/10.1088/0305-4470/36/38/309} {\bibfield  {journal} {\bibinfo  {journal} {Journal of Physics A: Mathematical and General}\ }\textbf {\bibinfo {volume} {36}},\ \bibinfo {pages} {9915} (\bibinfo {year} {2003})}\BibitemShut {NoStop}%
\bibitem [{\citenamefont {Luque}\ and\ \citenamefont {Thibon}(2006)}]{luque2006algebraic}%
  \BibitemOpen
  \bibfield  {author} {\bibinfo {author} {\bibfnamefont {J.-G.}\ \bibnamefont {Luque}}\ and\ \bibinfo {author} {\bibfnamefont {J.-Y.}\ \bibnamefont {Thibon}},\ }\bibfield  {title} {\bibinfo {title} {Algebraic invariants of five qubits},\ }\href {https://doi.org/10.1088/0305-4470/39/2/007} {\bibfield  {journal} {\bibinfo  {journal} {Journal of physics A: mathematical and general}\ }\textbf {\bibinfo {volume} {39}},\ \bibinfo {pages} {371} (\bibinfo {year} {2006})}\BibitemShut {NoStop}%
\bibitem [{\citenamefont {Li}\ and\ \citenamefont {Li}(2012{\natexlab{a}})}]{li2012slocc}%
  \BibitemOpen
  \bibfield  {author} {\bibinfo {author} {\bibfnamefont {X.}~\bibnamefont {Li}}\ and\ \bibinfo {author} {\bibfnamefont {D.}~\bibnamefont {Li}},\ }\bibfield  {title} {\bibinfo {title} {Slocc determinant invariants of order $2^{n/2}$ for even n qubits},\ }\href {https://doi.org/10.1088/1751-8113/45/7/075308} {\bibfield  {journal} {\bibinfo  {journal} {Journal of Physics A: Mathematical and Theoretical}\ }\textbf {\bibinfo {volume} {45}},\ \bibinfo {pages} {075308} (\bibinfo {year} {2012}{\natexlab{a}})}\BibitemShut {NoStop}%
\bibitem [{\citenamefont {Li}\ and\ \citenamefont {Li}(2012{\natexlab{b}})}]{li2012method}%
  \BibitemOpen
  \bibfield  {author} {\bibinfo {author} {\bibfnamefont {X.}~\bibnamefont {Li}}\ and\ \bibinfo {author} {\bibfnamefont {D.}~\bibnamefont {Li}},\ }\bibfield  {title} {\bibinfo {title} {Method for classifying multiqubit states via the rank of the coefficient matrix and its application to four-qubit states},\ }\href {https://doi.org/10.1103/PhysRevA.86.042332} {\bibfield  {journal} {\bibinfo  {journal} {Physical Review A—Atomic, Molecular, and Optical Physics}\ }\textbf {\bibinfo {volume} {86}},\ \bibinfo {pages} {042332} (\bibinfo {year} {2012}{\natexlab{b}})}\BibitemShut {NoStop}%
\bibitem [{\citenamefont {Li}\ and\ \citenamefont {Li}(2012{\natexlab{c}})}]{li2012classification}%
  \BibitemOpen
  \bibfield  {author} {\bibinfo {author} {\bibfnamefont {X.}~\bibnamefont {Li}}\ and\ \bibinfo {author} {\bibfnamefont {D.}~\bibnamefont {Li}},\ }\bibfield  {title} {\bibinfo {title} {Classification of general n-qubit states under stochastic local operations and classical communication in terms of the rank of coefficient matrix},\ }\href {https://doi.org/10.1103/PhysRevLett.108.180502} {\bibfield  {journal} {\bibinfo  {journal} {Physical review letters}\ }\textbf {\bibinfo {volume} {108}},\ \bibinfo {pages} {180502} (\bibinfo {year} {2012}{\natexlab{c}})}\BibitemShut {NoStop}%
\bibitem [{\citenamefont {Chen}\ and\ \citenamefont {Chen}(2006)}]{chen2006range}%
  \BibitemOpen
  \bibfield  {author} {\bibinfo {author} {\bibfnamefont {L.}~\bibnamefont {Chen}}\ and\ \bibinfo {author} {\bibfnamefont {Y.-X.}\ \bibnamefont {Chen}},\ }\bibfield  {title} {\bibinfo {title} {Range criterion and classification of true entanglement in a 2$\times$ m$\times$ n system},\ }\href {https://doi.org/10.1103/PhysRevA.73.052310} {\bibfield  {journal} {\bibinfo  {journal} {Physical Review A—Atomic, Molecular, and Optical Physics}\ }\textbf {\bibinfo {volume} {73}},\ \bibinfo {pages} {052310} (\bibinfo {year} {2006})}\BibitemShut {NoStop}%
\bibitem [{\citenamefont {Carteret}\ and\ \citenamefont {Sudbery}(2000)}]{carteret2000local}%
  \BibitemOpen
  \bibfield  {author} {\bibinfo {author} {\bibfnamefont {H.}~\bibnamefont {Carteret}}\ and\ \bibinfo {author} {\bibfnamefont {A.}~\bibnamefont {Sudbery}},\ }\bibfield  {title} {\bibinfo {title} {Local symmetry properties of pure three-qubit states},\ }\href {https://doi.org/10.1088/0305-4470/33/28/303} {\bibfield  {journal} {\bibinfo  {journal} {Journal of Physics A: Mathematical and General}\ }\textbf {\bibinfo {volume} {33}},\ \bibinfo {pages} {4981} (\bibinfo {year} {2000})}\BibitemShut {NoStop}%
\bibitem [{\citenamefont {Walck}\ and\ \citenamefont {Lyons}(2007)}]{walck2007maximum}%
  \BibitemOpen
  \bibfield  {author} {\bibinfo {author} {\bibfnamefont {S.~N.}\ \bibnamefont {Walck}}\ and\ \bibinfo {author} {\bibfnamefont {D.~W.}\ \bibnamefont {Lyons}},\ }\bibfield  {title} {\bibinfo {title} {Maximum stabilizer dimension for nonproduct states},\ }\href {https://doi.org/10.1103/PhysRevA.76.022303} {\bibfield  {journal} {\bibinfo  {journal} {Physical Review A—Atomic, Molecular, and Optical Physics}\ }\textbf {\bibinfo {volume} {76}},\ \bibinfo {pages} {022303} (\bibinfo {year} {2007})}\BibitemShut {NoStop}%
\bibitem [{\citenamefont {Gour}\ \emph {et~al.}(2017)\citenamefont {Gour}, \citenamefont {Kraus},\ and\ \citenamefont {Wallach}}]{gour2017almost}%
  \BibitemOpen
  \bibfield  {author} {\bibinfo {author} {\bibfnamefont {G.}~\bibnamefont {Gour}}, \bibinfo {author} {\bibfnamefont {B.}~\bibnamefont {Kraus}},\ and\ \bibinfo {author} {\bibfnamefont {N.~R.}\ \bibnamefont {Wallach}},\ }\bibfield  {title} {\bibinfo {title} {Almost all multipartite qubit quantum states have trivial stabilizer},\ }\bibfield  {journal} {\bibinfo  {journal} {Journal of Mathematical Physics}\ }\textbf {\bibinfo {volume} {58}},\ \href {https://doi.org/10.1063/1.5003015} {10.1063/1.5003015} (\bibinfo {year} {2017})\BibitemShut {NoStop}%
\bibitem [{\citenamefont {S{\l}owik}\ \emph {et~al.}(2020)\citenamefont {S{\l}owik}, \citenamefont {Hebenstreit}, \citenamefont {Kraus},\ and\ \citenamefont {Sawicki}}]{slowik2020link}%
  \BibitemOpen
  \bibfield  {author} {\bibinfo {author} {\bibfnamefont {O.}~\bibnamefont {S{\l}owik}}, \bibinfo {author} {\bibfnamefont {M.}~\bibnamefont {Hebenstreit}}, \bibinfo {author} {\bibfnamefont {B.}~\bibnamefont {Kraus}},\ and\ \bibinfo {author} {\bibfnamefont {A.}~\bibnamefont {Sawicki}},\ }\bibfield  {title} {\bibinfo {title} {A link between symmetries of critical states and the structure of slocc classes in multipartite systems},\ }\href {https://doi.org/10.22331/q-2020-07-20-300} {\bibfield  {journal} {\bibinfo  {journal} {Quantum}\ }\textbf {\bibinfo {volume} {4}},\ \bibinfo {pages} {300} (\bibinfo {year} {2020})}\BibitemShut {NoStop}%
\bibitem [{\citenamefont {Qi}\ and\ \citenamefont {Ranard}(2019)}]{qi2019determining}%
  \BibitemOpen
  \bibfield  {author} {\bibinfo {author} {\bibfnamefont {X.-L.}\ \bibnamefont {Qi}}\ and\ \bibinfo {author} {\bibfnamefont {D.}~\bibnamefont {Ranard}},\ }\bibfield  {title} {\bibinfo {title} {Determining a local hamiltonian from a single eigenstate},\ }\href {https://doi.org/10.22331/q-2019-07-08-159} {\bibfield  {journal} {\bibinfo  {journal} {Quantum}\ }\textbf {\bibinfo {volume} {3}},\ \bibinfo {pages} {159} (\bibinfo {year} {2019})}\BibitemShut {NoStop}%
\bibitem [{\citenamefont {Chertkov}\ and\ \citenamefont {Clark}(2018)}]{chertkov2018computational}%
  \BibitemOpen
  \bibfield  {author} {\bibinfo {author} {\bibfnamefont {E.}~\bibnamefont {Chertkov}}\ and\ \bibinfo {author} {\bibfnamefont {B.~K.}\ \bibnamefont {Clark}},\ }\bibfield  {title} {\bibinfo {title} {Computational inverse method for constructing spaces of quantum models from wave functions},\ }\href {https://doi.org/10.1103/PhysRevX.8.031029} {\bibfield  {journal} {\bibinfo  {journal} {Physical Review X}\ }\textbf {\bibinfo {volume} {8}},\ \bibinfo {pages} {031029} (\bibinfo {year} {2018})}\BibitemShut {NoStop}%
\bibitem [{\citenamefont {Petrovich}\ \emph {et~al.}(2024)\citenamefont {Petrovich}, \citenamefont {Rossignoli},\ and\ \citenamefont {Canosa}}]{petrovich2024covariance}%
  \BibitemOpen
  \bibfield  {author} {\bibinfo {author} {\bibfnamefont {F.}~\bibnamefont {Petrovich}}, \bibinfo {author} {\bibfnamefont {R.}~\bibnamefont {Rossignoli}},\ and\ \bibinfo {author} {\bibfnamefont {N.}~\bibnamefont {Canosa}},\ }\bibfield  {title} {\bibinfo {title} {Covariance-based method for eigenstate factorization and generalized singlets},\ }\href {https://doi.org/10.1103/PhysRevA.110.052213} {\bibfield  {journal} {\bibinfo  {journal} {Physical Review A}\ }\textbf {\bibinfo {volume} {110}},\ \bibinfo {pages} {052213} (\bibinfo {year} {2024})}\BibitemShut {NoStop}%
\bibitem [{\citenamefont {Petrovich}\ and\ \citenamefont {Rossignoli}(2025)}]{petrovich2025conserved}%
  \BibitemOpen
  \bibfield  {author} {\bibinfo {author} {\bibfnamefont {F.}~\bibnamefont {Petrovich}}\ and\ \bibinfo {author} {\bibfnamefont {R.}~\bibnamefont {Rossignoli}},\ }\bibfield  {title} {\bibinfo {title} {Conserved operators and exact conditions for pair condensation},\ }\href {https://doi.org/10.1103/qbwh-c3gp} {\bibfield  {journal} {\bibinfo  {journal} {Physical Review A}\ }\textbf {\bibinfo {volume} {112}},\ \bibinfo {pages} {032214} (\bibinfo {year} {2025})}\BibitemShut {NoStop}%
\bibitem [{\citenamefont {Barnum}\ \emph {et~al.}(2004)\citenamefont {Barnum}, \citenamefont {Knill}, \citenamefont {Ortiz}, \citenamefont {Somma},\ and\ \citenamefont {Viola}}]{barnum2004subsystem}%
  \BibitemOpen
  \bibfield  {author} {\bibinfo {author} {\bibfnamefont {H.}~\bibnamefont {Barnum}}, \bibinfo {author} {\bibfnamefont {E.}~\bibnamefont {Knill}}, \bibinfo {author} {\bibfnamefont {G.}~\bibnamefont {Ortiz}}, \bibinfo {author} {\bibfnamefont {R.}~\bibnamefont {Somma}},\ and\ \bibinfo {author} {\bibfnamefont {L.}~\bibnamefont {Viola}},\ }\bibfield  {title} {\bibinfo {title} {A subsystem-independent generalization of entanglement},\ }\href {https://doi.org/10.1103/PhysRevLett.92.107902} {\bibfield  {journal} {\bibinfo  {journal} {Physical Review Letters}\ }\textbf {\bibinfo {volume} {92}},\ \bibinfo {pages} {107902} (\bibinfo {year} {2004})}\BibitemShut {NoStop}%
\bibitem [{\citenamefont {Perelomov}(1977)}]{perelomov1977generalized}%
  \BibitemOpen
  \bibfield  {author} {\bibinfo {author} {\bibfnamefont {A.~d.~M.}\ \bibnamefont {Perelomov}},\ }\bibfield  {title} {\bibinfo {title} {Generalized coherent states and some of their applications},\ }\href {https://doi.org/10.1070/PU1977v020n09ABEH005459} {\bibfield  {journal} {\bibinfo  {journal} {Soviet Physics Uspekhi}\ }\textbf {\bibinfo {volume} {20}},\ \bibinfo {pages} {703} (\bibinfo {year} {1977})}\BibitemShut {NoStop}%
\bibitem [{\citenamefont {Gilmore}(1974)}]{gilmore1974properties}%
  \BibitemOpen
  \bibfield  {author} {\bibinfo {author} {\bibfnamefont {R.}~\bibnamefont {Gilmore}},\ }\bibfield  {title} {\bibinfo {title} {On the properties of coherent states},\ }\href {https://rmf.smf.mx/ojs/index.php/rmf/article/view/1046} {\bibfield  {journal} {\bibinfo  {journal} {Revista Mexicana de F{\'\i}sica}\ }\textbf {\bibinfo {volume} {23}},\ \bibinfo {pages} {143} (\bibinfo {year} {1974})}\BibitemShut {NoStop}%
\bibitem [{\citenamefont {Zhang}\ \emph {et~al.}(1990)\citenamefont {Zhang}, \citenamefont {Gilmore} \emph {et~al.}}]{zhang1990coherent}%
  \BibitemOpen
  \bibfield  {author} {\bibinfo {author} {\bibfnamefont {W.-M.}\ \bibnamefont {Zhang}}, \bibinfo {author} {\bibfnamefont {R.}~\bibnamefont {Gilmore}}, \emph {et~al.},\ }\bibfield  {title} {\bibinfo {title} {Coherent states: Theory and some applications},\ }\href {https://doi.org/10.1103/RevModPhys.62.867} {\bibfield  {journal} {\bibinfo  {journal} {Reviews of Modern Physics}\ }\textbf {\bibinfo {volume} {62}},\ \bibinfo {pages} {867} (\bibinfo {year} {1990})}\BibitemShut {NoStop}%
\bibitem [{\citenamefont {Provost}\ and\ \citenamefont {Vall{\'e}e}(1980)}]{provost1980riemannian}%
  \BibitemOpen
  \bibfield  {author} {\bibinfo {author} {\bibfnamefont {J.-P.}\ \bibnamefont {Provost}}\ and\ \bibinfo {author} {\bibfnamefont {G.}~\bibnamefont {Vall{\'e}e}},\ }\bibfield  {title} {\bibinfo {title} {Riemannian structure on manifolds of quantum states},\ }\href {https://doi.org/10.1007/BF02193559} {\bibfield  {journal} {\bibinfo  {journal} {Communications in Mathematical Physics}\ }\textbf {\bibinfo {volume} {76}},\ \bibinfo {pages} {289} (\bibinfo {year} {1980})}\BibitemShut {NoStop}%
\bibitem [{\citenamefont {Wierichs}\ \emph {et~al.}(2023)\citenamefont {Wierichs}, \citenamefont {East}, \citenamefont {Larocca}, \citenamefont {Cerezo},\ and\ \citenamefont {Killoran}}]{wierichs2023symmetric}%
  \BibitemOpen
  \bibfield  {author} {\bibinfo {author} {\bibfnamefont {D.}~\bibnamefont {Wierichs}}, \bibinfo {author} {\bibfnamefont {R.~D.}\ \bibnamefont {East}}, \bibinfo {author} {\bibfnamefont {M.}~\bibnamefont {Larocca}}, \bibinfo {author} {\bibfnamefont {M.}~\bibnamefont {Cerezo}},\ and\ \bibinfo {author} {\bibfnamefont {N.}~\bibnamefont {Killoran}},\ }\bibfield  {title} {\bibinfo {title} {Symmetric derivatives of parametrized quantum circuits},\ }\bibfield  {journal} {\bibinfo  {journal} {arXiv preprint arXiv:2312.06752}\ }\href {https://doi.org/10.48550/arXiv.2312.06752} {10.48550/arXiv.2312.06752} (\bibinfo {year} {2023})\BibitemShut {NoStop}%
\bibitem [{\citenamefont {Hall}(2013)}]{hall2013lie}%
  \BibitemOpen
  \bibfield  {author} {\bibinfo {author} {\bibfnamefont {B.~C.}\ \bibnamefont {Hall}},\ }\href@noop {} {\emph {\bibinfo {title} {Lie groups, Lie algebras, and representations}}}\ (\bibinfo  {publisher} {Springer},\ \bibinfo {year} {2013})\BibitemShut {NoStop}%
\bibitem [{\citenamefont {Perez-Garcia}\ \emph {et~al.}(2007)\citenamefont {Perez-Garcia}, \citenamefont {Verstraete}, \citenamefont {Wolf},\ and\ \citenamefont {Cirac}}]{perez2006matrix}%
  \BibitemOpen
  \bibfield  {author} {\bibinfo {author} {\bibfnamefont {D.}~\bibnamefont {Perez-Garcia}}, \bibinfo {author} {\bibfnamefont {F.}~\bibnamefont {Verstraete}}, \bibinfo {author} {\bibfnamefont {M.~M.}\ \bibnamefont {Wolf}},\ and\ \bibinfo {author} {\bibfnamefont {J.~I.}\ \bibnamefont {Cirac}},\ }\bibfield  {title} {\bibinfo {title} {Matrix product state representations},\ }\href {https://doi.org/10.5555/2011832.2011833} {\bibfield  {journal} {\bibinfo  {journal} {Quantum Info. Comput.}\ }\textbf {\bibinfo {volume} {7}},\ \bibinfo {pages} {401–430} (\bibinfo {year} {2007})}\BibitemShut {NoStop}%
\bibitem [{\citenamefont {Cerezo}\ \emph {et~al.}(2016)\citenamefont {Cerezo}, \citenamefont {Rossignoli},\ and\ \citenamefont {Canosa}}]{cerezo2016factorization}%
  \BibitemOpen
  \bibfield  {author} {\bibinfo {author} {\bibfnamefont {M.}~\bibnamefont {Cerezo}}, \bibinfo {author} {\bibfnamefont {R.}~\bibnamefont {Rossignoli}},\ and\ \bibinfo {author} {\bibfnamefont {N.}~\bibnamefont {Canosa}},\ }\bibfield  {title} {\bibinfo {title} {Factorization in spin systems under general fields and separable ground-state engineering},\ }\href {https://doi.org/10.1103/PhysRevA.94.042335} {\bibfield  {journal} {\bibinfo  {journal} {Physical Review A}\ }\textbf {\bibinfo {volume} {94}},\ \bibinfo {pages} {042335} (\bibinfo {year} {2016})}\BibitemShut {NoStop}%
\bibitem [{\citenamefont {Shiraishi}\ and\ \citenamefont {Mori}(2017)}]{shiraishi2017systematic}%
  \BibitemOpen
  \bibfield  {author} {\bibinfo {author} {\bibfnamefont {N.}~\bibnamefont {Shiraishi}}\ and\ \bibinfo {author} {\bibfnamefont {T.}~\bibnamefont {Mori}},\ }\bibfield  {title} {\bibinfo {title} {Systematic construction of counterexamples to the eigenstate thermalization hypothesis},\ }\href {https://doi.org/10.1103/PhysRevLett.119.030601} {\bibfield  {journal} {\bibinfo  {journal} {Physical review letters}\ }\textbf {\bibinfo {volume} {119}},\ \bibinfo {pages} {030601} (\bibinfo {year} {2017})}\BibitemShut {NoStop}%
\bibitem [{\citenamefont {Leifer}\ \emph {et~al.}(2004)\citenamefont {Leifer}, \citenamefont {Linden},\ and\ \citenamefont {Winter}}]{leifer2004measuring}%
  \BibitemOpen
  \bibfield  {author} {\bibinfo {author} {\bibfnamefont {M.~S.}\ \bibnamefont {Leifer}}, \bibinfo {author} {\bibfnamefont {N.}~\bibnamefont {Linden}},\ and\ \bibinfo {author} {\bibfnamefont {A.}~\bibnamefont {Winter}},\ }\bibfield  {title} {\bibinfo {title} {Measuring polynomial invariants of multiparty quantum states},\ }\href {https://doi.org/10.1103/PhysRevA.69.052304} {\bibfield  {journal} {\bibinfo  {journal} {Physical Review A}\ }\textbf {\bibinfo {volume} {69}},\ \bibinfo {pages} {052304} (\bibinfo {year} {2004})}\BibitemShut {NoStop}%
\bibitem [{\citenamefont {Yan}(2026)}]{yan2026strong}%
  \BibitemOpen
  \bibfield  {author} {\bibinfo {author} {\bibfnamefont {H.}~\bibnamefont {Yan}},\ }\bibfield  {title} {\bibinfo {title} {Strong symmetry from two-point correlations},\ }\bibfield  {journal} {\bibinfo  {journal} {arXiv preprint arXiv:2609.23591}\ }\href {https://doi.org/10.48550/arXiv.2609.23591} {10.48550/arXiv.2609.23591} (\bibinfo {year} {2026})\BibitemShut {NoStop}%
\bibitem [{\citenamefont {L{\'e}vay}\ and\ \citenamefont {Holweck}(2015)}]{levay2015embedding}%
  \BibitemOpen
  \bibfield  {author} {\bibinfo {author} {\bibfnamefont {P.}~\bibnamefont {L{\'e}vay}}\ and\ \bibinfo {author} {\bibfnamefont {F.}~\bibnamefont {Holweck}},\ }\bibfield  {title} {\bibinfo {title} {Embedding qubits into fermionic fock space: peculiarities of the four-qubit case},\ }\href {https://doi.org/10.1103/PhysRevD.91.125029} {\bibfield  {journal} {\bibinfo  {journal} {Physical Review D}\ }\textbf {\bibinfo {volume} {91}},\ \bibinfo {pages} {125029} (\bibinfo {year} {2015})}\BibitemShut {NoStop}%
\bibitem [{\citenamefont {Humphreys}(1975)}]{humphreys1975linear}%
  \BibitemOpen
  \bibfield  {author} {\bibinfo {author} {\bibfnamefont {J.~E.}\ \bibnamefont {Humphreys}},\ }\href@noop {} {\emph {\bibinfo {title} {Linear algebraic groups}}},\ Vol.~\bibinfo {volume} {21}\ (\bibinfo  {publisher} {Springer New York},\ \bibinfo {year} {1975})\BibitemShut {NoStop}%
\end{thebibliography}
\end{document}